\documentclass[runningheads,envcountsame]{llncs}
\usepackage{graphicx}
\usepackage{cite}

\usepackage[dvipsnames]{xcolor}

\usepackage{amsmath}
\usepackage{amssymb}
\usepackage{relsize}
\usepackage{bussproofs} \EnableBpAbbreviations
\usepackage{hhline}
\usepackage{adjustbox}

\usepackage{multicol}

\usepackage{tabularx}

\usepackage[noxy]{virginialake}
\usepackage{url}
\usepackage{paralist}
\usepackage{listliketab}
\usepackage{colortbl}
\usepackage{colonequals}
\usepackage[all]{xy}
\usepackage{tikz}
\usetikzlibrary{positioning,arrows,calc,shapes.geometric}
\tikzset{ modal/.style={>=stealth,shorten >=1pt,shorten <=1pt,auto,node distance=0.7cm, semithick}, world/.style={circle,draw,minimum size=0.2cm}, point/.style={circle,draw,inner sep=0.5mm,fill=black},point1/.style={circle,draw,inner sep=0.5mm}, reflexive above/.style={->,loop,looseness=7,in=120,out=60}, reflexive below/.style={->,loop,looseness=7,in=240,out=300}, reflexive left/.style={->,loop,looseness=7,in=150,out=210}, reflexive right/.style={->,loop,looseness=7,in=30,out=330}, itria/.style={draw,isosceles triangle,shape border rotate=270,yshift=-1.45cm, minimum height = 15mm}, empty/.style={inner sep=0.0mm } }

\usepackage[hidelinks]{hyperref}
\usepackage{thmtools}
\usepackage{thm-restate}

\renewcommand{\emptyset}{\varnothing}
\renewcommand{\phi}{\varphi}

\renewcommand{\epsilon}{\varepsilon}
\newcommand{\leqpos}{\leq_\ell}

\renewcommand{\L}{{\sf L}}

\newcommand{\K}{{\sf K}}
\newcommand{\D}{{\sf D}}
\newcommand{\T}{{\sf T}}

\newcommand{\Kfour}{{\sf K4}}
\newcommand{\Kfive}{{\sf K5}}
\newcommand{\KDfive}{{\sf KD5}}
\newcommand{\Kfourfive}{{\sf K45}}
\newcommand{\KDfourfive}{{\sf KD45}}
\newcommand{\KBfive}{{\sf KB5}}
\newcommand{\Sfour}{{\sf S4}}
\newcommand{\Sfive}{{\sf S5}}

\newcommand{\layKfive}{{\sf L.K5}}
\newcommand{\layKDfive}{{\sf L.KD5}}
\newcommand{\layKfourfive}{{\sf L.K45}}
\newcommand{\layKDfourfive}{{\sf L.KD45}}
\newcommand{\layKBfive}{{\sf L.KB5}}
\newcommand{\laySfive}{{\sf L.S5}}
\newcommand{\layL}{{\sf L.L}}

\newcommand{\logic}{{\sf L}}

\newcommand{\vlfill}[1]{{\{#1\}}}

\newcommand{\idTnd}{\mathsf{id_P}}
\newcommand{\idTop}{\mathsf{id_{\top}}}

\newcommand{\ruled}{{\sf d}}
\newcommand{\ruledtrunk}{{\sf d}_t'}
\newcommand{\doubledrule}{{\sf dd}}

\newcommand{\rulet}{{\sf t}}

\newcommand{\ce}{\colonequals}
\newcommand{\cce}{\coloncolonequals}

\newcommand{\calM}{\mathcal{M}}
\newcommand{\calN}{\mathcal{N}}
\newcommand{\calI}{\mathcal{I}}
\newcommand{\calJ}{\mathcal{J}}

\newcommand{\calG}{\mathcal{G}}

\newcommand{\calGcrown}{\mathcal{G}_c}
\newcommand{\calH}{\mathcal{H}}

\newcommand{\spdisj}{\mathbin{\varovee}}
\newcommand{\spconj}{\mathbin{\varowedge}}
\newcommand{\bigspdisj}{\mathop{\mathlarger{\mathlarger{\mathlarger{\mathlarger{\varovee}}}}}\limits}
\newcommand{\bigspconj}{\mathop{\mathlarger{\mathlarger{\mathlarger{\mathlarger{\varowedge}}}}}\limits}

\newcommand{\Lab}{\mathit{Lab}}
\newcommand{\Lit}{{\sf Lit}}
\newcommand{\LitDis}{{\sf LitDis}}
\newcommand{\Prop}{{\sf Pr}}
\newcommand{\mfal}[3]{{#1}\mathord{\downarrow}_{#3}{#2}}

\newbool{arxiv}
\booltrue{arxiv} 

\begin{document}

\title{Extensions of $\Kfive$: Proof Theory and\\ Uniform Lyndon Interpolation}
\titlerunning{Extensions of $\Kfive$: Proof Theory and Uniform Lyndon Interpolation}
%
\author{Iris van der Giessen\inst{1}\thanks{Supported by a UKRI
Future Leaders Fellowship, ‘Structure vs Invariants in Proofs’, project reference MR/S035540/1.} \and
Raheleh Jalali\inst{2}\thanks{Acknowledges the support of the Netherlands Organization for Scientific Research under grant 639.073.807 and the Czech Science Foundation Grant No. 22-06414L.}\orcidID{0000-0002-3321-8087} \and
Roman Kuznets\inst{3}\thanks{Supported by the Austrian Science Fund (FWF) ByzDEL project (P33600).}\orcidID{0000-0001-5894-8724}}
\authorrunning{I. van der Giessen et al.}
%
\institute{University of Birmingham,  UK, \email{i.vandergiessen@bham.ac.uk} \and
Utrecht University,  Netherlands and Czech Academy of Sciences, Prague, Czechia, \email{rahele.jalali@gmail.com} \and
TU Wien,  Austria, \email{roman@logic.at}
}
\maketitle              
\begin{abstract}
We introduce a Gentzen-style framework, called \emph{layered sequent calculi}, for modal logic~$\Kfive$ and its extensions~$\KDfive$, $\Kfourfive$, $\KDfourfive$, $\KBfive$,~and~$\Sfive$ with the goal to investigate the uniform Lyndon interpolation property (ULIP), which implies both the uniform interpolation property and the Lyndon interpolation property. We obtain complexity-optimal decision procedures for all logics and 
present a constructive proof  of  the ULIP for $\Kfive$, which to the best of our knowledge, is the first such syntactic proof.
To prove that the interpolant is correct, we use model-theoretic methods, especially bisimulation modulo literals.
\end{abstract}

\section{Introduction}
The uniform interpolation property~(UIP) is an important property of a logic. It strengthens  the Craig interpolation property~(CIP) by making  interpolants  depend on only one formula of an implication, either the premise or conclusion. A lot of work has gone into proving the UIP, and it is shown to be useful in various areas of computer science, including knowledge representation~\cite{Koopmann15PhD} and description logics~\cite{lutz}. 
Early results on the UIP in modal logic include positive results proved  semantically for  logics~$\mathsf{GL}$~and~$\mathsf{K}$ (independently in \cite{Shavrukov93,Ghilardi95,Visser96}) and negative results for logics~$\Sfour$~\cite{GhilardiZawadowski95} and $\Kfour$~\cite{Bilkova06PhD}.
A proof-theoretic method to prove the UIP was first proposed in~\cite{Pitts92JSL} for intuitionistic propositional logic and later adapted to modal logics, such as~$\mathsf{K}$~and~$\mathsf{T}$ in~\cite{Bilkova06PhD}. 
A general proof-theoretic method of proving the UIP  for many classical and intuitionistic (non-)normal modal logics and substructural (modal) logics based on the form of their sequent-calculi rules was developed in the series of papers~\cite{Iemhoff19b,Semianalytic2,tabatabai2022}.\looseness=-1

Apart from the UIP, we are also interested in the uniform Lyndon interpolation property~(ULIP) that is a strengthening of the UIP in the sense that interpolants must respect the polarities of the propositional variables involved. Kurahashi~\cite{Kur20} first introduced this property and proved it for  several normal modal logics, by employing a semantic method using layered bisimulations. A sequent-based proof-theoretic method  was used in~\cite{tab2022uniform} to show the ULIP for several non-normal modal logics and conditional logics. 

Our long-term goal is to provide a general proof-theoretic method to (re)prove the UIP for modal logics via multisequent calculi (i.e.,~nested sequents, hypersequents, labelled hypersequents, etc.). Unlike many other ways of proving interpolation, the proof-theoretic treatment is constructive in that it additionally yields an algorithm for constructing uniform interpolants. Towards this goal, we build on the modular treatment of multicomponent calculi to prove the CIP for modal and intermediate logics in~\cite{FitKuz15APAL,Kuz16Proof,Kuz18APAL,KuzLell18,KuzLel21}. First steps have been made by reproving the UIP for modal logics~$\K$,~$\D$,~and~$\T$ via nested sequents~\cite{van2021} and for~$\Sfive$ via  hypersequents~\cite{GieJalKuz21arx,vdGiessen22PhD}, the first time this is proved proof-theoretically   for~$\Sfive$. \looseness=-1

In this paper, we focus on logics~$\Kfive$, $\KDfive$, $\Kfourfive$, $\KDfourfive$, $\KBfive$,~and~$\Sfive$. The ULIP for these logics was derived in~\cite[Prop.~3]{Kur20} from the logics' local tabularity~\cite{NaglyThomason85} and  Lyndon interpolation property~(LIP)~\cite{Kuz16}.

Towards a modular proof-theoretic treatment, we introduce a new form of multisequent calculi for these logics that we call \emph{layered sequent calculi}, the structure of which is inspired by the structure of the Kripke frames for the concerned logics from~\cite{nag}. For $\Sfive$, this results in  standard hypersequents~\cite{Min71PSIM,Pot83JSL,Avr96}. For~$\Kfive$ and $\KDfive$, the presented calculi are similar to grafted hypersequent calculi in~\cite{kuznetsLellmann16} but without explicit weakening. Other, less related, proof systems include analytic cut-free sequent systems for $\Kfive$ and $\KDfive$ \cite{Takano01}, cut-free sequent calculi for $\Kfourfive$ and $\KDfourfive$ \cite{Svartz}, and nested sequent calculi for modal logics \cite{Brue09AML}. 

The layered sequent calculi introduced in this paper adopt a strong version of termination that only relies on a local loop-check based on saturation. For all concerned logics, this yields a decision procedure that runs in co-NP time, which is, therefore, optimal~\cite{HalReg}. We provide a semantic completeness proof via a countermodel construction from failed proof search. \looseness=-1

Finally,  layered sequents are used  to provide the first proof-theoretic proof of the ULIP for $\Kfive$. The method is adapted from~\cite{vdGiessen22PhD,GieJalKuz21arx} in which the UIP is proved for $\Sfive$ based on hypersequents. We provide an algorithm to construct uniform Lyndon interpolants purely by syntactic means using the termination strategy of the proof search. To show the correctness of the constructed interpolants, we use model-theoretic techniques inspired by bisimulation quantification in the setting of uniform  Lyndon interpolation~\cite{Kur20}. 

\section{Preliminaries}\label{sec:preliminaries}
The language of modal logics consists of a set~$\mathsf{Pr}$ of countably many (\emph{propositional}) \emph{atoms}~$p, q, \ldots$, their \emph{negations}~$\overline{p}, \overline{q}, \ldots$, \emph{propositional connectives}~$\wedge$~and~$\vee$, \emph{boolean constants}~$\top$~and~$\bot$, and \emph{modal operators}~$\Box$~and~$\Diamond$. A \emph{literal}~$\ell$ is either an atom or its negation, and the set of all literals is denoted by~$\Lit$. We define \emph{modal formulas} in the usual way and denote them by lowercase Greek letters $\phi, \psi, \ldots$. We define~$\overline{\phi}$ using the usual De~Morgan laws to push the negation inwards (in particular, $\overline{\overline{p}} \colonequals p$)  and $\phi \to \psi \ce \overline{\phi} \vee \psi$. We use uppercase Greek letters $\Gamma, \Delta, \ldots$ to refer to finite \emph{multisets} of formulas. We write $\Gamma,\Delta$ to mean $\Gamma \cup \Delta$ and $\Gamma, \phi$ to mean $\Gamma \cup \{ \phi \}$. The set of literals of a formula $\phi$, denoted $\Lit (\phi)$, is defined recursively: $\Lit(\top)=\Lit(\bot)=\varnothing$, $\Lit (\ell)=\ell$ for $\ell \in \Lit$, $\Lit (\phi \wedge \psi)= \Lit (\phi \vee \psi)= \Lit (\phi) \cup \Lit (\psi)$, and $\Lit (\Box \phi)= \Lit (\Diamond \phi)=\Lit (\phi)$.

We consider extensions of~$\Kfive$ with any combination of axioms~$\mathsf{4}$, $\mathsf{d}$, $\mathsf{b}$,~and~$\mathsf{t}$ (Table~\ref{table:axiom}). Several of the 16~combinations coincide, resulting in 6~logics: $\Kfive$, $\KDfive$, $\Kfourfive$, $\KDfourfive$, $\KBfive$,~and~$\Sfive$ (Table \ref{table:logics}). Throughout the paper, we assume $\logic \in \{\Kfive, \KDfive, \Kfourfive, \KDfourfive, \KBfive, \Sfive \}$ and write $\vdash_\logic \phi$ if{f} $\phi \in \logic$.

\begin{definition}[Logic $\Kfive$]\label{Dfn:K5}
\emph{Modal logic~$\Kfive$} is axiomatized by the classical tautologies, axioms\/~$\sf k$~and\/~$\sf 5$, and rules modus ponens $($\/from $\phi$ and $\phi \to \psi$ infer  $\psi)$ and necessitation $($\/from $\phi$ infer\/ $\Box \phi)$.
\end{definition}

Throughout the paper we employ the semantics of Kripke frames and models. 

\begin{definition}[Kripke semantics] \label{def:Kripke}
A \emph{Kripke frame} is a pair\/~$(W,R)$ where $W$~is a nonempty set of \emph{worlds} and $R \subseteq W \times W$ a binary relation.
A \emph{Kripke model} is a triple\/~$(W, R, V)$ where\/ $(W,R)$ is a Kripke frame and $V\colon \mathsf{Pr} \to \mathcal{P}(W)$ is a \emph{valuation function}. A formula~$\phi$ is defined to be \emph{true} at a world~$w$ in a model~$\mathcal{M}=(W,R,V)$, denoted $\mathcal{M}, w \vDash \phi$, as follows: $\calM,w\vDash\top$, $\calM,w\nvDash\bot$ and
\begin{center}
\begin{tabular}{l @{\extracolsep{1em}} c l}
$\calM,w\vDash p$ & if{f} & $w\in V(p)$\\
$\calM,w\vDash \overline{p}$ & if{f} & $w\notin V(p)$\\
$\calM,w\vDash\varphi\land\psi$ & if{f} & $\calM,w\vDash\varphi$ and $\calM,w\vDash\psi$\\
$\calM,w\vDash\varphi\lor\psi$ & if{f} & $\calM,w\vDash\varphi$ or $\calM,w\vDash\psi$\\
$\calM,w\vDash\Box\varphi$ & if{f} & for all $v \in W$ such that $wRv$, $\mathcal M,v\vDash\varphi$\\
$\calM,w\vDash\Diamond\varphi$ & if{f} & there exists $v \in W$ such that $wRv$ and $\calM,v\vDash\varphi$.\\
\end{tabular}
\end{center}
Formula~$\phi$ is \emph{valid in~$\calM = (W,R,V)$}, denoted\/ $\calM \vDash \phi$, if{f} for all $w \in W$, $\calM,w \vDash \phi$.  We call\/ $\varnothing \ne C \subseteq W$ a \emph{cluster} $($\/in~$\mathcal{M})$ if{f} $C \times C \subseteq R$, i.e.,~the relation~$R$ is \emph{total} on~$C$. We write $w R C$ if{f} $w R v$ for all $v \in C$.
\end{definition}

\begin{table}[t]
\begin{center}
\setlength{\tabcolsep}{10pt}
\footnotesize{\begin{tabular}{ |c| c| c| }
\hline
Axiom &  Formula & Frame condition\\
\hline
$\sf k$ & $\Box(\phi \to \psi) \to (\Box \phi \to \Box \psi)$ & none\\
\hline
$\sf 5$ & $\Diamond \phi \to \Box \Diamond \phi$ & Euclidean: $w R v \wedge w R u \Rightarrow v R u$\\
\hhline{|=|=|=|}
$\sf 4$ & $\Box \phi \to \Box \Box \phi$ & transitive: $w R v \wedge v R u \Rightarrow w R u$\\
\hline
$\sf d$ & $\Box \phi \to \Diamond \phi$ & serial: $\forall w \exists v (w R v)$\\
\hline
$\sf b$ & $\phi \to \Box \Diamond \phi$  & symmetric:  $w R v \Rightarrow v R w$\\
\hline
$\sf t$ & $\Box \phi \to \phi$ & reflexive: $\forall w (w R w)$\\
\hline
\end{tabular}}\vspace{0.5em}
\caption{\small{Modal axioms and their corresponding frame conditions.\vspace{-1em}}}\label{table:axiom}
\end{center}
\end{table}

\begin{table}[t]
\begin{center}
\setlength{\tabcolsep}{12pt}
\footnotesize{\begin{tabular}{ |c| c| c| }
\hline
Logic \sf L &  Axiomatization &  Class of \sf L-frames $(W, R)$\\
\hline
$\Kfive$ & Definition \ref{Dfn:K5} & $W=\{\rho\}$ or $W=\{\rho\} \sqcup C$\\
\hline
$\KDfive$ & $\Kfive + \sf d$ & $W=\{\rho\} \sqcup C$\\
\hline
$\Kfourfive$ & $\Kfive + \sf 4$ & $W=\{\rho\}$ or ($W=\{\rho\} \sqcup C$ and $\rho R C$)\\
\hline
$\KDfourfive$ & $\Kfive + \sf d +\sf 4$ & $W=\{\rho\} \sqcup C$ and $\rho R C$\\
\hline
$\KBfive$ & $\Kfive + \sf b$  &  $W=\{\rho\}$ or $W=C$\\
\hline
$\Sfive$ & $\Kfive + \sf t$ & $W=C$\\
\hline
\end{tabular}}\vspace{0.5em}
\parbox{0.9\textwidth}{\caption{\small{Semantics for extensions of $\Kfive$ (see \cite{nag,pietruszczak2020}).
Everywhere not $\rho R \rho$\\ for the root~$\rho$, set $C$ is a  finite cluster, and $\sqcup$ denotes disjoint union.\vspace{-1em}}}
}
\label{table:logics}
\end{center}
\end{table}

We work with specific classes of Kripke models sound and complete w.r.t. the logics. The respective frame conditions for the logic $\logic$, called \emph{$\logic$-frames}, are defined in Table~\ref{table:logics}. A model $(W,R,V)$ is an \emph{$\sf L$-model} if{f} $(W,R)$ is an $\sf{L}$-frame. Table~\ref{table:logics} is a refinement of Theorem \ref{thm: nag}, particularly shown for $\Kfourfive$, $\KDfourfive$, and $\KBfive$ in~\cite{pietruszczak2020}. More precisely, we consider rooted frames and completeness w.r.t.~the root, i.e., $\vdash_\logic \phi$ if{f} for all $\logic$-models $\calM$ with root $\rho$, $\calM, \rho \vDash \phi$ (we often denote the if-condition as $\vDash_\logic \phi$). For each logic, this follows from easy bisimulation arguments.\looseness=-1

\begin{theorem}[\cite{nag}] \label{thm: nag} Any normal modal logic containing\/~$\Kfive$ is sound and complete w.r.t.~a class of finite Euclidean Kripke frames\/ $(W, R)$ of  one of the following forms:\/ $(a)$~$W = \{\rho\}$ consists of a singleton root and $R=\emptyset$, $(b)$ the whole~$W$ is a cluster (any world can be considered its root), or $(c)$ $W \backslash\{\rho\}$ is a cluster for a (unique) root $\rho \in W$ such that $\rho R w$ for some $w \in W \backslash\{\rho\}$ while not $\rho R \rho$.
\end{theorem}

\begin{definition}[UIP and ULIP] \label{Dfn: ULIP}
A logic\/~$\L$ has the \emph{uniform interpolation property~(UIP)} if{f} for any formula $\phi$
and $p \in \Prop$ there is a
formula $\forall p \phi$
such that 
\begin{compactenum}[(1)]
\item\label{UIP:i} $\Lit(\forall p  \phi) \subseteq \Lit(\phi) \setminus \{ p, \overline{p} \}$,
\item\label{UIP:ii}
$\vdash_\L \forall p \phi \to \phi$, and
\item\label{UIP:iii}
$\vdash_\L \psi \rightarrow \phi  \text{ implies }  \vdash_\logic \psi \rightarrow \forall p \phi$ for any formula $\psi$ with $p,\overline{p} \notin \Lit(\psi)$.
\end{compactenum}

A logic\/~$\L$ has the \emph{uniform Lyndon interpolation property~(ULIP)}~\textup{\cite{Kur20,tab2022uniform}}
if{f} for any formula $\phi$ and $\ell \in \Lit$, there is a formula\/~$\forall \ell \phi$ such that 
\begin{compactenum}[(i)]
\item\label{ULIP:i} $\Lit(\forall \ell  \phi) \subseteq \Lit(\phi) \setminus \{ \ell \}$,
\item\label{ULIP:ii}
$\vdash_\L \forall \ell \phi \to \phi$, and
\item\label{ULIP:iii}
$\vdash_\L \psi \rightarrow \phi  \text{ implies }  \vdash_\logic \psi \rightarrow \forall \ell \phi$ for any formula $\psi$ with $\ell \notin \Lit(\psi)$.
\end{compactenum}
We call $\forall p \phi$ $(\forall \ell \phi)$ the \emph{uniform (Lyndon) interpolant of $\phi$ w.r.t.~atom $p$ (literal~$\ell$)}.
\end{definition}

\noindent
{These are often called  \emph{pre-interpolants} as opposed to their dual \emph{post-interpolants} that, in classical logic, can be defined as $\exists p \phi = \overline{\forall p \overline{\phi}}$ and $\exists \ell \phi = \overline{\forall \overline{\ell} \overline{\phi}}$ (see, e.g., \cite{tab2022uniform,Bilkova06PhD,vdGiessen22PhD,Kur20} for more explanations).} 

\begin{theorem}
If a logic\/~$\L$ has the ULIP, 
then it also has the UIP.
\end{theorem}
\begin{proof}
We define a uniform interpolant of $\phi$ w.r.t.~atom $p$ as
a uniform Lyndon interpolant~$\forall p \forall \overline{p} \phi$ of~$\forall \overline{p} \phi$ w.r.t.~literal $p$. We need to demonstrate conditions LIP\eqref{UIP:i}--\eqref{UIP:iii} from Def.~\ref{Dfn: ULIP}.
First, it follows from ULIP\eqref{ULIP:i} that $\Lit(\forall p \forall \overline{p} \phi) \subseteq \Lit(\forall \overline{p} \phi) \setminus \{p\} \subseteq \Lit(\phi) \setminus \{p, \overline{p}\}$. 
Second, $\vdash_{\L} \forall p \forall \overline{p} \phi \to \forall\overline{p} \phi$ and $\vdash_{\L} \forall \overline{p} \phi \to \phi $ by ULIP\eqref{ULIP:ii}, hence, $\vdash_{\L} \forall p \forall \overline{p} \phi \to \phi$. 
Finally, if $\vdash_{\L} \psi \to \phi$ where $p, \overline{p} \notin \Lit(\psi)$, then  by 
ULIP\eqref{ULIP:iii}, $\vdash_{\L} \psi \to \forall \overline{p} \phi$ as $\overline{p} \notin \Lit(\psi)$ and $\vdash_{\L} \psi \to \forall p \forall \overline{p} \phi$ as $p \notin \Lit(\psi)$. \qed
\end{proof}

\section{Layered Sequents}

\begin{definition}[Layered sequents]\label{def:RootedLayeredHypersequents}
A \emph{layered sequent} is a generalized one-sided sequent of the form
\begin{equation}
    \label{eq:labelled}
\mathcal{G}= \Gamma_1, \ldots, \Gamma_n, [\Sigma_1], \ldots, [\Sigma_m],  [[\Pi_1]] , \ldots,  [[\Pi_k]] 
\end{equation}
where~$\Gamma_i, \Sigma_i, \Pi_i$ are finite  multisets of formulas, $n,m,k \geq 0$, and  if \mbox{$k\geq 1$}, then $m \geq 1$.  A layered sequent is an\/ $\logic$-sequent if{f} it satisfies the conditions in the rightmost column of Table~\ref{table:modal}. 
Each~$\Sigma_i$, each~$\Pi_i$, and\/~$\bigcup_i \Gamma_i$ is called a \emph{sequent component} of~$\calG$.
The \emph{formula interpretation} of a layered sequent~$\calG$  above is: \looseness=-1
\[
\iota(\calG ) = \bigvee\nolimits_{i=1}^n \big(\bigvee\Gamma_i\big) \vee \bigvee\nolimits_{i=1}^m \Box \big(\bigvee \Sigma_i\big) \vee \bigvee\nolimits_{i=1}^k \Box \Box \big(\bigvee \Pi_i\big).
\]
\end{definition}

Layered sequents are denoted by $\calG$ and $\calH$. The structure of a layered sequent can be viewed as at most two layers of hypersequents (\emph{$[\enspace]$-components} $\Sigma_i$ and \emph{$[[\enspace]]$-components} $\Pi_i$ forming the  first and second layer respectively) possibly nested on top of the sequent component~$\bigcup_i \Gamma_i$ as the root. Following the arboreal terminology from~\cite{kuznetsLellmann16}, the root is called the \emph{trunk} while $[\enspace]$- and $[[\enspace]]$-components form the \emph{crown}. Analogously to  nested sequents representing  tree-like Kripke models, the structure of   $\logic$-sequents is in line with the structure of $\logic$-models introduced in Sect.~\ref{sec:preliminaries}. We view sequents components as freely permutable, e.g., $[[\Pi_1]], \Gamma_1, [\Sigma_1], \Gamma_2$ and $\Gamma_1, \Gamma_2, [\Sigma_1], [[\Pi_1]]$ represent the same layered sequent.

\begin{table}[t]
\begin{center}
\footnotesize{\begin{tabular}{ | l | c c c c c c c c| c| c| }
\hline
\ Calculus \  &  \multicolumn{8}{c|}{Sequent rules} & Conditions on layered sequents \\
\hline
\ $\layKfive$ & \ $\Box_t$ & &  $\Diamond_t$ & &  $\Box_{c'}$ & & &  & $n \geq 1, \; m,k \geq 0 $ \\
\hline
\ $\layKDfive$ & \  $\Box_t$ & & $\Diamond_t$  & & $\Box_{c'}$ & $\ruled_t$ & & $\ruled_{c'}$ \  & $n \geq 1, \; m,k \geq 0 $ \\
\hline
\ $\layKfourfive$ & \ $\Box_t$ & & $\Diamond_t$ &  $\Box_{c}$ & & & & & $n \geq 1, \; m \geq 0, \; k=0 $  \\
\hline
\ $\layKDfourfive$ & \ $\Box_t$ &  & $\Diamond_t$ & $\Box_{c}$ &  & $\ruled_t$ &  $\ruled_c$ &   & $n \geq 1, \; m \geq 0, \; k=0 $    \\
\hline
\ $\layKBfive$ & & \ $\Box_{t'}$ & &  $\Box_{c}$ & & & & &  \ $n =0,  m \geq 2,  k=0 \;$ or $\; n =1,  m =0,  k=0 $  \ \\
\hline
\ $\laySfive$ & & & & $\Box_{c}$ & & & &  & $n =0, \; m \geq 1, \; k=0 $  \\
\hline
\end{tabular}}\vspace{0.5em}
\caption{\small{Layered sequent calculi $\layL$: in addition to explicitly stated rules, all  $\layL$ have axioms $\idTnd$ and $\idTop$ and rules $\vee$, $\wedge$, $\Diamond_c$, and $\rulet$ (see Fig.~\ref{fig:rules}). Note that the rules of system~$\layL$ may only be applied to $\logic$-sequents.\vspace{-1em}}}\label{table:modal}
\end{center}
\end{table}
\begin{figure}[t]
\centering
  \fbox{\parbox{0.98\textwidth}{
   \begin{center}
      $\vlinf{\idTnd}{}{\mathcal{G}\vlfill{p,\overline p}}{}$
      \quad
      $\vlinf{\idTop}{}{\mathcal{G}\vlfill{\top}}{}$
      \quad
      $\vliinf{\land}{}{\mathcal{G}\vlfill{\varphi \land \psi}}{\mathcal{G}\vlfill{\varphi \land \psi,\varphi}}{\mathcal{G}\vlfill{\varphi \land \psi,\psi}}$       \end{center}  
      \begin{center}    
      $\vlinf{\lor}{}{\mathcal{G}\vlfill{\varphi \lor \psi}}{\mathcal{G}\vlfill{\varphi \lor \psi,  \varphi,\psi}}$
        \quad
      $\vlinf{\Box_t}{}{\mathcal{G}, \Box \varphi}{\mathcal{G}, \Box \varphi,[\varphi]}$
      \quad
      $\vlinf{\Box_{t'}}{}{\Sigma, \Box \varphi}{[\Sigma, \Box \varphi],[\varphi]}$
      \quad
      $\vlinf{\Diamond_t}{}{\mathcal{G}, \Diamond \varphi, [\Sigma]}{\mathcal{G}, \Diamond \varphi, [\Sigma, \phi]}$
      \end{center}
     \begin{center}
      $\vlinf{\Box_c}{}{\mathcal{G}, \llbracket \Sigma, \Box \varphi \rrbracket}{\mathcal{G}, \llbracket \Sigma, \Box \varphi \rrbracket, [\phi]}$
      \quad
      $\vlinf{\Box_{c'}}{}{\mathcal{G}, \llbracket \Sigma, \Box \varphi \rrbracket}{\mathcal{G}, \llbracket \Sigma, \Box \varphi \rrbracket,  [[\phi]]}$
      \quad
      $\vlinf{\Diamond_c}{}{\mathcal{G}, \llbracket \Sigma, \Diamond \varphi \rrbracket, \llparenthesis \Pi \rrparenthesis}{\mathcal{G}, \llbracket \Sigma, \Diamond \varphi \rrbracket, \llparenthesis \Pi, \phi \rrparenthesis}$
      \end{center}   
      \begin{center}
     $\vlinf{\ruled_t}{}{\mathcal{G}, \Diamond \varphi}{\mathcal{G}, \Diamond \varphi,[\varphi]}$
      \quad
      $\vlinf{\ruled_{c}}{}{\mathcal{G}, \llbracket \Sigma, \Diamond \varphi \rrbracket}{\mathcal{G}, \llbracket \Sigma, \Diamond \varphi \rrbracket, [\phi]}$
      \quad
     $\vlinf{\ruled_{c'}}{}{\mathcal{G}, \llbracket \Sigma, \Diamond \varphi \rrbracket}{\mathcal{G}, \llbracket \Sigma, \Diamond \varphi \rrbracket, [[\phi]]}$
\quad
          $\vlinf{\rulet}{}{\mathcal{G}, \llbracket \Sigma, \Diamond \varphi \rrbracket}{\mathcal{G}, \llbracket \Sigma, \Diamond \varphi , \phi \rrbracket}$
      \end{center}
  }}
\caption{Layered sequent rules: brackets $\llbracket \enspace \rrbracket$ and $ \llparenthesis \enspace \rrparenthesis$ range over both $[\enspace]$ and $[[\enspace]]$.} 
\label{fig:rules}
\end{figure}

{\begin{remark}The layered calculi presented here generalize grafted hypersequents of~\cite{kuznetsLellmann16} and, hence, similarly combine features of hypersequents and nested sequents. In particular, layered sequents are generally neither pure hypersequents (except for the case of $\Sfive$) nor bounded-depth nested sequents. The latter is due to the fact that the defining property of nested sequents is the tree structure of the sequent components, whereas the crown components of a layered sequent form a cluster. Although formally grafted hypersequents are defined with one layer only,  this syntactic choice is more of a syntactic sugar than a real distinction. Indeed, the close relationship of one-layer grafted hypersequents for $\Kfive$~and~$\KDfive$ in~\cite{kuznetsLellmann16} to the two-layer layered sequents presented here clearly manifests itself when translating grafted hypersequents into the prefixed-tableau format  (see grafted tableau system for $\Kfive$~\cite[Sect.~6]{kuznetsLellmann16}). There prefixes for the crown are separated into two types, limbs and twigs, which match the separation into $[\enspace]$-~and~$[[\enspace]]$-components.\end{remark}}

For a layered sequent~\eqref{eq:labelled}, we assign labels to the components as follows: 
the trunk is labeled~$\bullet$, $[ \enspace]$-components get distinct  labels $\bullet 1, \bullet 2, \dots$, and  $[[ \enspace]]$-components get  distinct labels~$1,2, \dots$. 
We let $\sigma, \tau, \dots$ range over these labels. The set of labels is denoted $\Lab(\calG)$ and $\sigma \in \calG$ means $\sigma \in \Lab(\calG)$. We write $\sigma : \phi\in \calG$ (or $\sigma : \phi$ if no confusion occurs) when a formula $\phi$ occurs in a sequent component of $\calG$ labeled by $\sigma$.

\begin{example}
$\mathcal{G}= \phi, \psi, [\chi], [\xi], [[\theta]]$ is a layered sequent with the trunk and three crown components: two $[\enspace]$-components and one $[[\enspace]]$-component. Since it has both the trunk and a $[[ \enspace]]$-component, it can only be a $\Kfive$- or $\KDfive$-sequent. A corresponding labeled sequent is $\mathcal{G}= \phi_\bullet, \psi_\bullet,  [\chi]_{\bullet 1}, [\xi]_{\bullet 2}, [[\theta]]_1$, with the set $\Lab(\mathcal{G}) = \{\bullet, \bullet 1, \bullet 2, 1\}$ of four labels. Similarly, for the $\KBfive$/$\Sfive$-sequent $\cal H= [\sigma], [\delta]$, a corresponding labeled sequent is $\mathcal{H}= [\sigma]_{\bullet 1}, [\delta]_{\bullet 2}$  with $\Lab(\mathcal{H}) = \{\bullet 1, \bullet 2\}$.
\end{example}

We sometimes use \emph{unary contexts}, i.e.,  layered sequents with exactly one \emph{hole}, denoted  $\{ \enspace \}$. Such contexts are denoted by $\calG\{ \enspace \}$. The insertion $\calG\{\Gamma\}$ of a finite multiset $\Gamma$ into $\calG\{ \enspace \}$ is obtained by replacing $\{ \enspace \}$ with $\Gamma$. The hole $\{ \enspace \}$ in a component $\sigma$  can also be labeled $\calG\{ \enspace \}_\sigma$. We use the notations $\llbracket \enspace \rrbracket$ and $ \llparenthesis \enspace \rrparenthesis$ to refer to either of $[\enspace]$ or $[[\enspace]]$. 

Using Fig.~\ref{fig:rules} and the middle column of Table~\ref{table:modal}, we define  layered sequent calculi $\layKfive$, $\layKDfive$, $\layKfourfive$, $\layKDfourfive$, $\layKBfive$,~and~$\laySfive$, where~$\layL$ is the calculus for the logic~$\logic$.  Following the terminology from~\cite{kuznetsLellmann16}, we split all modal rules into  \emph{trunk rules} (subscript $t$) and \emph{crown rules} (subscript $c$) depending on the position of the \emph{principal} formula. We write $\vdash_\layL \calG$ if{f} $\calG$ is derivable in $\layL$.

\begin{definition}[Saturation]
Labeled formula $\sigma: \phi \in \calG$ is \emph{saturated for\/~$\layL$}  if{f}
 \begin{compactitem}
    \item $\phi$ equals $p$ or $\overline{p}$ for an atom $p$, or equals\/ $\bot$, or equals\/ $\top$;   
        \item $\phi = \phi_1 \wedge \phi_2$ and $\sigma : \phi_i \in \calG$ for some $i$;
        \item $\phi = \phi_1 \vee \phi_2$ and both $\sigma : \phi_1 \in \calG$ and $\sigma : \phi_2 \in \calG$; 
        \item $\phi=\Box \phi'$, the unique rule applicable to $\sigma: \Box \phi'$ in\/ $\layL$ is either\/ $\Box_t$ or\/ $\Box_c$ $($i.e., a rule creating  a\/ $[\enspace]$-component$)$, and\/ $\bullet i : \phi' \in \calG$ for some $i$;
        \item $\phi=\Box \phi'$, the unique rule applicable to $\sigma: \Box \phi'$ in\/ $\layL$ is\/ $\Box_{c'}$ $($i.e., a rule creating  a\/ $[[\enspace]]$-component$)$, and $i : \phi' \in \calG$ for some $i$.
    \end{compactitem}
    In addition, we define for any label $\sigma$ and formula $\phi$:
    \begin{compactitem}
        \item $\sigma: \Diamond \phi$ is \emph{saturated w.r.t.}\/~$\bullet \in \Lab(\calG)$;
    \item $\sigma: \Diamond \phi$ is \emph{saturated w.r.t.~a label}\/ $\bullet i \in \Lab(\calG)$ if{f}\/ $\bullet i : \phi \in \calG$;
    \item $\sigma: \Diamond \phi$  is \emph{saturated w.r.t.~a label} $i\in \Lab(\calG)$ if{f} $\sigma = \bullet$ or $i : \phi \in \calG$;
        \item $\sigma : \Diamond \phi$ is\/ \emph{$\ruled_t$-saturated} if{f} $\sigma \ne \bullet$ or\/ $\bullet i : \phi \in \calG$ for some $i$; 
        \item
        $\sigma : \Diamond \phi$  is\/ \emph{$\ruled_c$-saturated} if{f} $\sigma = \bullet$ or\/ $\bullet i : \phi \in \calG$ for some $i$; 
        \item $\sigma : \Diamond \phi$  is\/ \emph{$\ruled'_c$-saturated} if{f} $\sigma = \bullet$ or $i: \phi \in \calG$ for some~$i$.
            \end{compactitem}
$\calG$ is \emph{propositionally saturated} if{f} all\/ $\vee$- and\/ $\wedge$-formulas are saturated in $\calG$.
$\logic$-sequent $\calG$ is\/ \emph{$\logic$-saturated} if{f} $a)$~each non-$\Diamond$ formula is saturated, $b)$~each $\sigma : \Diamond \phi$ is saturated w.r.t.~every label in $\Lab(\calG)$, $c)$~each $\sigma : \Diamond \phi$ is\/ ${\sf d}$-saturated whenever\/ ${\sf d} \in \layL \cap \{\ruled_t, \ruled_c, \ruled_{c'} \}$, and $d)$~$\calG$ is not of the from $\calH\{\top\}$ or $\calH\{q, \overline{q}\}$ for some $q \in \Prop$.\looseness=-1
\end{definition}

\begin{theorem}\label{thm:termination_proof_search}
 Proof search in\/~${\layL}$ modulo saturation terminates and provides an optimal-complexity decision algorithm, i.e., runs in co-NP time.
\end{theorem}
\begin{proof}
    Given a proof search of layered sequent $\calG$, for each layered sequent $\calH$ in this proof search, consider its labeled formulas as a set $F_\calH=\{\sigma : \phi \mid \sigma : \phi \in \calH \}$. Let $s$ be the number of subformulas occurring in $\calG$ and $N$ be the number of sequent components in $\calG$. Since we only apply rules (that do not equal $\idTnd$ or $\idTop$) to non-saturated sequents, sets $F_\calH$ will grow for each premise. Going bottom-up in the proof search, at most $s$ labels of the form~$\bullet i$ and at most $s$ labels of the form $i$ can be created, and each label can have at most~$s$ formulas. Therefore, the cardinality of sets $F_\calH$ are bounded by  $s(N+s + s)$, which is polynomial in the size of $F_\calG$. Hence, the proof search terminates modulo saturation. Moreover, since each added labeled formula is linear in the size $F_\calG$ and the non-deterministic branching in the proof search is bounded by $(N+s+s)s(N+s+s)$, 
    again a polynomial in the size of $F_\calG$, this algorithm is co-NP, i.e., provides an optimal decision procedure for the logic.\qed
\end{proof}

\begin{definition}[Interpretations] \label{Dfn:  interpretation}
    An \emph{interpretation of an $\logic$-sequent~$\calG$ into an $\logic$-model $\calM = (W,R,V)$} is a function~$\calI : \Lab(\calG) \to W$ such that the following conditions  apply whenever the respective type of labels exists in~$\calG$: 
       
       \begin{compactenum}
        \item\label{interp:1} 
        $\calI(\bullet) = \rho$, where $\rho$ is the root of $\calM$;
        \item\label{interp:2} 
        $\calI(\bullet)R\;\calI(\bullet i)$ for each label of the form\/ $\bullet i \in \Lab(\calG)$;
        \item\label{interp:3}  
        $\calI(\bullet i) R\; \calI(j)$ and $\calI(j) R\; \calI(\bullet i)$ for all labels of the form\/ $\bullet i$ and $j$ in $\Lab(\calG)$;
        \item\label{interp:4}
        Not $\calI(\bullet) R\; \calI(j)$ for any label of the form $j \in \Lab(\calG)$. 
    \end{compactenum}
\end{definition}
Note that none of the conditions~\eqref{interp:1}--\eqref{interp:4} apply to layered $\Sfive$-sequents. 

\begin{definition}[Sequent semantics]
    For any given interpretation $\calI$ of an\/ $\logic$-sequent~$\calG$ into an\/ $\logic$-model $\calM$, 
\begin{center}
$    \calM, \calI \vDash \calG \qquad\text{ if{f} }\qquad \calM, \calI(\sigma) \vDash \phi \text{ for some } \sigma : \phi \in \calG.$
\end{center}
    $\calG$ is \emph{valid} in\/ $\logic$, denoted\/ $\vDash_\logic \calG$, if{f} $\calM, \calI \vDash \calG$ for all\/ $\logic$-models $\calM$ and interpretations~$\calI$ of $\calG$ into $\calM$. We omit\/~$\logic$ and $\cal M$ when  clear from the context.
\end{definition}

The proof of the following theorem is  based on a countermodel construction (with more standard parts of the proof relegated to the Appendix):
\begin{restatable}[Soundness and completeness]{theorem}{soundcomplete}
\label{thm:soundcomplete}
    For any\/ $\logic$-sequent $\calG$,
    \[
    \vdash_\layL \calG 
    \qquad\Longleftrightarrow\qquad
    \vDash_\logic \iota(\calG)
    \qquad\Longleftrightarrow\qquad
    \vDash_\logic \calG.
    \]
\end{restatable}
\begin{proof}
   We show a cycle of implications. The left-to-middle implication, i.e., that $\vdash_\layL \calG \Longrightarrow \ \vDash_\logic \iota(\calG)$, can be proved by induction on the  $\layL$-derivation  of $\calG$.

For the middle-to-right implication, i.e.,   $\vDash_\logic \iota(\calG) \Longrightarrow \ \vDash_\logic \calG$, 
let $\calG$ be a sequent of  form~\eqref{eq:labelled}.
We prove that $\calM, \calI \nvDash \calG$ implies $\calM, \calI(\bullet) \nvDash \iota(\calG)$ (if $n=0$, use $1$ in place of~$\bullet$). By definition, $\calI(\bullet)$ is the root of $\calM$. If $\calM, \calI \nvDash \calG$, then $\calI(\bullet) \nvDash \phi$ for all $\phi \in \bigcup\nolimits_{i=1}^n \Gamma_i$, for each $1 \leq i \leq m$ we have $\calI(\bullet i) \nvDash \psi$ for all $\psi \in \Sigma_i$, and for each $1 \leq i \leq k$ we have $\calI(i) \nvDash \chi$ for all $\chi \in  \Pi_i$. By Def.~\ref{Dfn: interpretation}, 
    in case $k \geq 1$ label~$\bullet 1$ is in $\calG$ and $\calI(\bullet)R\calI(\bullet 1)R\calI(i)$ for each $1 \leq i \leq k$. Therefore $\calM, \calI(\bullet) \nvDash \iota(\calG)$.

Finally, we prove  the right-to-left implication by contraposition using a countermodel construction: 
from a failed proof search of $\calG$, 
construct an $\logic$-model refuting  $\calG$ from~\eqref{eq:labelled}. In a failed proof-search tree (Theorem~\ref{thm:termination_proof_search}), since $\nvdash_\layL \calG$, at least one saturated leaf\looseness=-1
\vspace{-0.15em}
\[
\calG'=\Gamma', [\Sigma'_1], \ldots, [\Sigma'_m],  [\Sigma''_1], \ldots, [\Sigma''_{m'}],[[\Pi'_1]] , \ldots,  [[\Pi'_k]], [[\Pi''_{1}]] , \ldots,  [[\Pi''_{k'}]],\vspace{-0.15em}
\]
is such that $\bigcup_i \Gamma_i \subseteq \Gamma'$, $\Sigma_i \subseteq \Sigma'_i$, and $\Pi_i \subseteq \Pi'_i$ (or for $\KBfive$, if $\calG = \Gamma$, 
then $\calG' = \Gamma'$ for $\Gamma \subseteq \Gamma'$ or $[\Sigma], [\Sigma_1], \dots, [\Sigma_m]$ with $\Gamma \subseteq \Sigma$). Define $\calM = (W,R,V)$: 
\vspace{-1.5pt}
    \begin{gather*}
        W = \Lab(\calG'),\qquad\qquad V(p) = \{ \sigma \mid \sigma : \overline{p} \in \calG' \},\\
        R = \{ (\bullet, \bullet i) \mid \bullet i \in \Lab(\calG')\} \cup \{ (\sigma, \tau) \mid \sigma, \tau \in \Lab(\calG'), \sigma, \tau \neq \bullet \}.
    \end{gather*}
    Since $\calG'$ is saturated,  $\calM$ is an $\logic$-model. Taking $\calI$ of $\calG$ into $\calM$ as the identity function (or $\calI(\bullet) = 1$ in case of $\KBfive$), we have $\calM, \calI \nvDash \calG$ as desired. \qed
\end{proof}

\section{Uniform Lyndon Interpolation}
\begin{definition}[Multiformulas]\label{def:multiformula}
The grammar
\[
\mho \cce \sigma : \varphi \mid (\mho \spconj \mho) \mid (\mho \spdisj \mho)
\]
defines \emph{multiformulas}, 
where $\sigma: \phi$ is a labeled formula. $\Lab(\mho)$ denotes the set of labels of\/ $\mho$. An \emph{interpretation $\calI$} of a layered sequent $\cal G$ into a model $\calM$ is called an \emph{interpretation of a multiformula\/ $\mho$ into $\calM$} if{f} $\Lab (\mho) \subseteq \Lab(\calG)$. If $\calI$~is an  interpretation of\/ $\mho$ into $\calM $, we define $\calM, \calI \vDash \mho$ as follows:\\ 
\begin{listliketab} 
    \storestyleof{compactenum} 
        \begin{tabular}{lll}
             &$\calM,\calI \vDash \sigma : \varphi$\;& if{f} \; $\calM, \calI(\sigma) \vDash \varphi$,\\
             &$\calM, \calI \vDash \mho_1 \spconj \mho_2$\; & if{f} \; $\calM, \calI \vDash \mho_1$ and $\calM, \calI \vDash \mho_2$,\\
             &$\calM, \calI \vDash \mho_1 \spdisj \mho_2$\;& if{f} \; $\calM, \calI \vDash \mho_i$ for at least one~$i=1,2$.
        \end{tabular} 
\end{listliketab}\\
Multiformulas\/ $\mho_1$ and\/ $\mho_2$ are said to be \emph{equivalent}, denoted\/ $\mho_1 \equiv_\logic \mho_2$, or simply\/ $\mho_1 \equiv \mho_2$, if{f} $\calM,\calI \vDash \mho_1 \Leftrightarrow \calM, \calI \vDash \mho_2$ for any interpretation $\calI$ of both\/ $\mho_1$~and\/~$\mho_2$ into an\/ $\logic$-model $\calM$.
\end{definition}

\begin{lemma}[\cite{Kuz18APAL}] \label{Lem: SDNF}
    Any multiformula\/ $\mho$ can be transformed into an equivalent one in SDNF $($SCNF\/$)$ as a\/ $\spdisj$-disjunction $(\spconj$-conjunction$)$ of\/ $\spconj$-conjunctions $(\spdisj$-disjunctions$)$ of labeled formulas $\sigma : \varphi$ such that each label of\/~$\mho$ occurs exactly once per conjunct (disjunct). 
\end{lemma}

\begin{definition}[Bisimilarity] \label{def: Bisimilarity}  
Let $\calM=(W,R,V)$ and $\calM'=(W',R',V')$ be models and $\ell \in \Lit$. We say $\cal M'$ is \emph{$\ell$-bisimilar} to $\cal M$, denoted $\mathcal{M'} \leq_{\ell} \mathcal{M}$ if{f} there is a nonempty binary relation $Z \subseteq W \times W'$, called an \emph{$\ell$-bisimulation} between~$\mathcal{M}$~and~$\mathcal{M'}$, such that the following hold for every $w \in W$ and $w' \in W'$:
\begin{compactdesc}
\item[\textup{literals}$_{\ell}$.]
if $w Z w'$, then $a)$ $\mathcal{M}, w \vDash q$ if{f}  $\mathcal{M}', w' \vDash q$ for all atoms $q \notin \{\ell, \overline{\ell}\}$ and\\
 $b)$ if $\mathcal{M}', w' \vDash \ell$, then $\mathcal{M}, w \vDash \ell$; 
\item[\textup{forth}.]
if $w Z w'$ and $w R v$, then there exists $v'\in W'$ such that $v Zv'$ and $w' R' v'$; 
\item[\textup{back}.]
if $w Z w'$ and $w' R' v'$, then there exists $v \in W$ such that $v Zv'$ and $w R v$.
\end{compactdesc}
$\calM$ and $\calM'$ are \emph{bisimilar}, denoted $\calM \sim \calM'$, if{f} there is a  relation $Z \neq \emptyset$ satisfying\/ \textbf{\textup{forth}} and\/ \textbf{\textup{back}}, as well as part $a)$ of\/ \textbf{\textup{literals}}$_{\ell}$ for any $p\in \Prop$, in which case $Z$  is called a \emph{bisimulation}. We write\/ $($similarly for\/ $\sim$ instead of\/ $\leq_{\ell})$:\looseness=-1
\begin{compactitem} 
\item $(\calM', w') \leq_{\ell} (\calM,w)$ if{f} there is an $\ell$-bisimulation $Z$, such that $wZw'$;
\item $(\calM', \calI') \leq_{\ell} (\calM,\calI)$ for functions $\calI : X \to W$ and $\calI' : X \to W'$ if{f} there is an $\ell$-bisimulation~$Z$ such that $ \calI(\sigma)\, Z\, \calI'(\sigma)$ for each $\sigma \in X$.
\end{compactitem}
\end{definition}
Note that $\leq_{\ell}$ is a  preorder and we have
$\calM' \leq_\ell \calM$ if{f} $\calM \leq_{\overline{\ell}} \calM'$. 
By analogy with~\cite[Theorem 2.20]{blackburn}, we have the following immediate observation, which additionally holds for multiformulas $\mho$ (see the Appendix for a proof): 

\begin{restatable}{lemma}{modaleq}\label{lem_modal_equivalence}
     Let $\calI$ and~$\calI'$ be interpretations of a layered sequent $\calG$  into models~$\calM$~and~$\calM'$ respectively.
     \begin{compactenum}
        \item 
        Let $\ell \notin \Lit(\calG)$. If $(\calM', \calI') \leq_\ell (\calM, \calI)$, then   $\calM, \calI \vDash \calG$ implies $\calM', \calI' \vDash \calG$.
        \item \label{lem_modal_equivalence:full}
        If $(\calM, \calI) \sim (\calM', \calI')$, then $\calM, \calI \vDash \calG$ if{f} $\calM', \calI' \vDash \calG$.
     \end{compactenum}
\end{restatable}

\begin{definition}[BLUIP] \label{def:BLUIP}
   Logic\/ $\logic$ is said to have  the \emph{bisimulation layered-sequent  uniform interpolation property (BLUIP)}  if{f}  for every  
    literal $\ell$ and every\/ $\logic$-sequent $\calG$, there is a multiformula $A_\ell (\calG)$, called \emph{BLU interpolant}, such that:\looseness=-1
\begin{compactenum}[(i)]
\item \label{BLUIP:1} 
$\Lit\bigl(A_\ell (\cal G)\bigr) \subseteq  \Lit(\cal G) \setminus$$\{\ell\}$ and $\Lab \bigl(A_\ell (\cal G)\bigr) \subseteq \Lab (\cal G)$;
\item \label{BLUIP:2}
for each interpretation~$\calI$ of\/~$\cal G$ into an\/ $\logic$-model~$\calM$,
\[
\calM, \calI \vDash 
A_\ell (\cal G) \quad\text{implies}\quad \calM, \calI \vDash \cal G;
\]
\item \label{BLUIP:3} 
for each\/ $\logic$-model~$\calM$ and  interpretation~$\calI$ of\/~$\cal G$ into~$\calM$, if 
$
\calM, \calI \nvDash
A_\ell (\cal G)$,
then there is an\/ $\logic$-model~$\calM'$ and interpretation~$\calI'$ of\/~$\cal G$ into~$\calM'$ such that $$(\calM',\calI') \leq_\ell (\calM, \calI) \; \text{and} \;
\calM', \calI' \nvDash \cal G.$$
\end{compactenum}
\end{definition}

\begin{lemma} \label{lem: BLUIP implies UIP}
The BLUIP for\/~$\logic$ implies the  ULIP for\/~$\logic$. 
\end{lemma}
\begin{proof}
     Let~$\forall \ell \phi = A_\ell (\phi)$. We prove the properties of Def.~\ref{Dfn: ULIP}.  
     Variable property is immediate.
     For Property \eqref{ULIP:ii}, assume~$\nvdash_\L 
     A_\ell(\phi) \rightarrow \phi$. By completeness, we have~$\calM, \rho \vDash 
     A_\ell(\phi)$ and~$\calM,\rho \nvDash \phi$ for some $\L$-model~$\calM$ with root $\rho$. As~$\rho$ is the root, it can be considered as an interpretation by Def.~\ref{Dfn:  interpretation}. By condition~\eqref{BLUIP:2} from Def.~\ref{def:BLUIP} we get a contradiction. For~\eqref{ULIP:iii}, let $\psi$ be a formula such that~$\ell \notin \Lit(\psi)$ and suppose~$\nvdash_\L \psi \rightarrow 
     A_\ell(\phi)$. So there is an $\L$-model~$\calM$ with root~$\rho$ such that $\calM, \rho \vDash \psi$ and~$\calM, \rho \nvDash 
     A_\ell(\phi)$. Again, $\rho$
     is treated as an interpretation, and by~\eqref{BLUIP:3} from Def.~\ref{def:BLUIP}, there is an $\L$-model~$\calM'$ with root~$\rho'$ such that $(\calM',\rho') 
     \leq_\ell (\calM,\rho)$ and $\calM',\rho' \nvDash \phi$.  By Lemma~\ref{lem_modal_equivalence}, $\calM', \rho' \vDash \psi$, hence $\nvdash_\L \psi \rightarrow \phi$ as desired.\looseness=-1 \qed
\end{proof}

To show that  calculus $\layKfive$   enjoys the BLUIP for $\Kfive$, we need two important ingredients: some model modifications that are closed under bisimulation and an algorithm to compute uniform Lyndon interpolants.
\begin{definition}[Copying]
\label{def:copying}
    Let $\calM=(W,R,V)$ be a\/ $\Kfive$-model with root $\rho$ and cluster $C$. Model $\calN' = (W \sqcup \{w_c \},R',V')$ is obtained by \emph{copying  $w\in C$} if{f} $R' = R \sqcup (\{w_c\}\times C) \sqcup (C \times \{w_c\}) \sqcup \{(\rho,w_c) \mid (\rho,w) \in R \} \sqcup \{(w_c,w_c) \}$, and $V'(p)= V(p) \sqcup \{w_c \mid w \in V(p) \}$ for any $p \in \Prop$. Model $\calN'' = (W \sqcup \{w_c \},R'',V')$ is obtained by \emph{copying~$w$ away from the root} if{f} $R'' = R' \setminus \{(\rho, w_c)\}$.
\end{definition}

\begin{lemma}\label{lem:copying}
    Let model $\calN$ be obtained by copying a world $w$ from a\/ $\Kfive$-model~$\calM$ (away from the root). Let $\calI \colon X \to \calM$ and $\calI'\colon X \to \calN$ be interpretations  such that for each $x \in X$, either $\calI(x) = \calI'(x)$ or $\calI(x)=w$ while $\calI'(x)= w_c$. Then, $\calN$ is a\/ $\Kfive$-model and $(\calM,\calI) \sim (\calN,\calI')$.
\end{lemma}

In the construction of  interpolants, we use the following rules $\ruledtrunk$ and $\doubledrule$ and sets $\calG_c$ and $\Box \Diamond \calG_c$ of formulas from the crown of $\calG$:
\begin{center}
    $\calG_c = \{ \phi \mid \sigma : \phi \in \calG, \sigma \neq \bullet\} \qquad
    \Box \Diamond \calG_c = \{ \Box \phi \mid \Box \phi \in \calG_c\} \sqcup \{ \Diamond \phi \mid \Diamond \phi \in \calG_c\}$
\end{center}

\begin{center}
$\vliinf{\ruledtrunk}{}{\Gamma}{\Gamma, \bigl[\{\psi \mid \Diamond \psi \in \Gamma \}\bigr]}{\Gamma, \Diamond \top} \ \ \text{ and } \ \ 
\vlinf{\doubledrule}{}{\mathcal{G}}{\mathcal{G}, \bigl[\{\psi \mid \Diamond \psi \in \calG \}\bigr], \bigl[\bigl[ \{ \chi \mid \Diamond \chi \in \calGcrown\}\bigr]\bigr]}$
\end{center}
Rule $\ruledtrunk$ shows similarities with rule $\ruled_t$ from logics $\KDfive$ and $\KDfourfive$, but is only applied in the absence of the crown. Rule $\ruledtrunk$ is sound for $\Kfive$ because it can be viewed as a composition of an (admissible) cut on $\Box \bot$ and~$\Diamond \top$ in the trunk, followed by $\Box_t$ in the left premise on $\Box \bot$ that creates the first crown  component (though  $\bot$ is dropped from it), which is populated using several $\Diamond_t$-rules for $\Diamond \psi \in \Gamma$. The label of this  crown component is always $\bullet 1$. 
Rule $\doubledrule$ provides extra information in the calculation of the uniform interpolant and is needed primarily for technical reasons.  We highlight the two new sequent components created by the last instance of $\doubledrule$ using special placeholder labels $\bullet \ruled$ and $\ruled$ for the respective brackets. These labels are purely for  readability purposes and revert to the standard  $\bullet j$ and $k$ labels after the next instance of $\doubledrule$. 

\begin{table}[t]
  \begin{center}
  \renewcommand*{\arraystretch}{1.4}
    \begin{tabular}{|l  l l |} \hline
     & $\calG$~matches \quad \quad \quad 
    & 
    $A_\ell(t, \Sigma_c; \calG)$~equals
    \\
    \hline
    1. & $\calG'\vlfill{\top}_\sigma$ 
    &
    $\sigma : \top$ 
    \\
    2. & $\calG'\vlfill{q,\overline{q}}_\sigma$ 
    &
    $\sigma : \top$ 
    \\
    3. & $\calG'\vlfill{\varphi \lor \psi}$ 
    &
    $A_\ell\bigl(t, \Sigma_c; \calG'\vlfill{\varphi \lor \psi, \varphi, \psi}\bigr)$
    \\
    4. & $\calG'\vlfill{\varphi \land \psi}$ 
    &
    $A_\ell\bigl(t, \Sigma_c; \calG'\vlfill{\varphi \land \psi, \varphi}\bigr) \spconj A_\ell\bigl(t, \Sigma_c; \calG'\vlfill{\varphi \land \psi, \psi}\bigr)$
    \\
    5. & $\calG', \Box \varphi$
    &
    $\bigspconj_{i=1}^h \left(\bullet : \Box \delta_i \spdisj \bigspdisj_{\tau \in \calG} \tau : \gamma_{i,\tau} \right)$
    \\
    & & where $j$~is the smallest integer such that $\bullet j \notin \calG$ and
    the SCNF  \\
    & & of  $A_\ell\bigl(t, \Sigma_c; \calG', \Box \varphi, [\varphi]_{\bullet j}\bigr)$ is $\bigspconj_{i=1}^h \left(\bullet j : \delta_i  \spdisj \bigspdisj_{\tau \in \calG} \tau : \gamma_{i,\tau} \right)$,
    \\
    6. & $\calG', \llbracket \Sigma , \Box \varphi \rrbracket_\sigma$
    &
    $\bigspconj_{i=1}^h \left(\sigma : \Box \delta_i \spdisj \bigspdisj_{\tau \in \calG} \tau : \gamma_{i,\tau} \right)$
    \\
    & & where $j$~is the smallest integer such that $j \notin \calG$ and
    the SCNF  \\
    & & of  $A_\ell\bigl(t, \Sigma_c; \calG'$, $\llbracket \Sigma , \Box \varphi \rrbracket_\sigma, [[\phi]]_j \bigr)$ is $\bigspconj_{i=1}^h \left(j : \delta_i  \spdisj \bigspdisj_{\tau \in \calG} \tau : \gamma_{i,\tau} \right)$,
    \\
    7. & $\calG', \Diamond \varphi,[\Sigma]$
    &
    $A_\ell\bigl(t, \Sigma_c; \calG', \Diamond \phi, [\Sigma, \phi] \bigr)$
    \\
    8. & $\calG', \llbracket \Sigma, \Diamond \phi \rrbracket$
    &
    $A_\ell\bigl(t, \Sigma_c; \calG', \llbracket \Sigma, \Diamond \phi, \phi \rrbracket\bigr)$
    \\
    9. & $\calG',\llbracket \Sigma, \Diamond \phi \rrbracket, \llparenthesis \Pi \rrparenthesis$
    &
    $A_\ell\bigl(t, \Sigma_c; \calG',\llbracket \Sigma, \Diamond \phi \rrbracket, \llparenthesis \Pi , \phi \rrparenthesis \bigr)$
    \\
    \hline
    \end{tabular}
    \end{center}
\caption{Recursive construction of~$A_\ell(t, \Sigma_c; \calG)$ for~$\calG$ that are not $\Kfive$-saturated. 
}
\label{table:Ap}
\end{table}

To compute a uniform Lyndon interpolant $\forall \ell \xi$ for a formula $\xi$, we first  compute a BLU interpolant  $A_\ell(0,\varnothing ; \xi_\bullet)$ by using the recursive function $A_\ell(t,\Sigma_c ; \calG)$ with three parameters we present below. The main parameter is a $\Kfive$-sequent~$\calG$, while the other two parameters are auxiliary:  $t \in \{0,1\}$ is a boolean variable such that $t=1$ guarantees that rule~$\doubledrule$ has been applied at least once for the case when $\calG$ contains  diamond formulas;
$\Sigma_c \subseteq \Box \Diamond \calG_c$ is a set of modal formulas that provides a bookkeeping strategy to prevent redundant applications of rule~$\doubledrule$. \looseness=-1

To calculate $A_\ell(t,\Sigma_c ; \calG)$ our algorithm makes a choice of which row from Table \ref{table:Ap} to apply by  trying each of the following steps in the specified order: 
\begin{compactenum}
    \item\label{step_initial}
    If possible, apply rows 1--2, i.e., stop and return $A_\ell(t, \Sigma_c;\calG) =\sigma : \top$.
    \item\label{step_prop}
    If 
    some formula $\varphi \lor \psi $  (resp.~$\varphi \land \psi$) from $\calG$ is not saturated, compute $A_\ell(t, \Sigma_c;\calG)$ according to row~3~(resp.~4)  applied to this formula.
    \item\label{step_modal}
    If some formula $\Box \varphi\in\calG$ is not saturated (resp. $\Diamond \varphi\in\calG$ is not saturated w.r.t.~$\sigma \in \calG$), compute $A_\ell(t, \Sigma_c;\calG)$ according to the unique respective  row among 5--9  applicable to  this formula (w.r.t.~$\sigma$).
    \item\label{step_drule}
    If Steps \ref{step_initial}--\ref{step_modal} do not apply, i.e., $\calG$ is saturated, proceed as follows:
    \begin{compactenum}
 \item\label{step_saturated_base}
    if $\calG$ has  no $\Diamond$-formulas, stop and return  $A_\ell(t, \Sigma_c;\calG) = \LitDis_\ell(\calG)$ where
    \begin{equation}\label{Ap_saturated_base}
            \LitDis_\ell(\calG) = 
            \bigspdisj_{\sigma :  \ell' \in \calG,  \ell' \in \Lit \setminus \{ \ell\}} \sigma : \ell'
        \end{equation}
            
        \item\label{step_saturated_dtrunk} 
        else, if $\calG=\Gamma$ consists of the trunk only, apply rule $\ruledtrunk$ as follows:
        \begin{multline}\label{Ap_saturated_dtrunk}
            \vspace{-10pt}   A_\ell(t, \Sigma_c;\Gamma) =\\ 
            \Bigl(\bullet : \Box \bot \spdisj  
            \bigspdisj_{i=1}^{h} \big( \bullet : \Diamond \delta_i  \spconj  \bullet : \gamma_{i} \big) 
            \Bigr) 
            \spconj 
            \bigl(\bullet : \Diamond \top \spdisj \LitDis_\ell(\Gamma)  \bigr)
        \end{multline}      
where the SDNF of $A_\ell\Bigl(0, \Sigma_c ;\quad \Gamma, \bigl[\{\psi \mid \Diamond \psi \in \Gamma \}\bigr]_{ \bullet 1}\Bigr)$ is
        \begin{equation}\label{Ap_saturated_SDNFdtrunk}
            \bigspdisj_{i=1}^{h} \Bigl( \bullet 1
            : \delta_i \spconj \bullet : \gamma_{i} \Bigr)
        \end{equation}
        \item\label{step_saturated_insufficeintbase}
        else, if $t=1$ and $\Box \Diamond \calG_c \subseteq \Sigma_c$, stop and return $A_\ell(t, \Sigma_c;\calG) = \LitDis_\ell(\calG)$. 
        \item\label{step_saturated_doubled}
        else, apply the rule $\doubledrule$ as follows (where w.l.o.g.~$\bullet 1 \in \calG)$: 
        \begin{equation}\label{Ap_saturated_doubled}
            A_\ell(t, \Sigma_c;\calG) = \bigspdisj_{i=1}^{h} \left(\bullet : \Diamond \delta_i  \spconj \bullet 1 : \Diamond \delta'_i \spconj \bigspconj_{ \tau \in \calG } \tau : \gamma_{i,\tau} \right) 
        \end{equation}
        where SDNF of $A_\ell\Bigl(1, \Box \Diamond \calG_c;\,\, \calG, \bigl[\{\psi \mid \Diamond \psi \in \calG \}\bigr]_{\bullet \ruled},\: \bigl[\bigl[ \{ \chi \mid \Diamond \chi \in \calGcrown \} \bigr]\bigr]_{\ruled}\Bigr)$ is\looseness=-1
        \begin{equation}\label{Ap_saturated_SDNFdoubled}
            \bigspdisj_{i=1}^{h} \left(\bullet \ruled: \delta_i  \spconj \ruled : \delta'_i \spconj \bigspconj_{ \tau \in \calG} \tau: \gamma_{i,\tau} \right)
        \end{equation}
    \end{compactenum}   
\end{compactenum}

The computation of the algorithm can be seen as a  proof search tree (extended with rules $\ruledtrunk$ and $\doubledrule$).  
In this proof search, call
$A_\ell(t,\Sigma_c ; \calG)$ is \textit{sufficient} (to be a BLU interpolant for~$\calG$) if each branch going up  from it either stops in Steps~\ref{step_initial}~or~\ref{step_saturated_base} or continues via Steps~\ref{step_saturated_dtrunk} or \ref{step_saturated_doubled}. Otherwise, it is \textit{insufficient}, if one of the branches stops in Step~\ref{step_saturated_insufficeintbase}, say, calculating $A_\ell(1,\Sigma_c ; \calH)$. In this  case,  $A_\ell(1,\Sigma_c ; \calH)$ is not generally a BLU interpolant for $\calH$, but these leaves provide enough information to find a BLU interpolant from some sequent down the proof search tree.\looseness=-1
\begin{example}
    Consider the layered sequent $\calG= \phi$ for $\phi =\overline{p} \vee \Diamond \Diamond (p \vee q)$.  We show how to construct $A_\ell(0, \emptyset;\phi)$ for $\ell=p$.  First, we compute the proof search tree decorated with  $(t, \Sigma_c)$ to the left of each line, according to the algorithm, using the following abbreviations  $\Gamma = \phi, \overline{p}, \Diamond \Diamond (p \vee q)$ and $\Sigma_1 = \Diamond (p \vee q), p \vee q, p, q$:
\begin{center}
 \adjustbox{max width=\textwidth}{
    \AxiomC{$(1, \{\Diamond (p \vee q)\}) \ \ \  
    {\Gamma, [\Sigma_1]_{\bullet1}, [\Diamond(p \vee q), p \vee q, p, q]_{\bullet \ruled}, [[p \vee q, p, q]]_{\ruled}}$} \RightLabel{\scriptsize{$\vee$} }
    \UnaryInfC{$(1, \{\Diamond (p \vee q)\}) \ \ \  
    \Gamma, [\Sigma_1]_{\bullet1}, [\Diamond(p \vee q), p \vee q]_{\bullet \ruled}, [[p \vee q]]_{\ruled}$} \RightLabel{\scriptsize{$\doubledrule$} } 
    \UnaryInfC{$(0, \emptyset) \ \ \ \Gamma, [\Diamond (p \vee q), p \vee q, p, q]_{\bullet 1}$} \RightLabel{\scriptsize{$\vee$} }
    \UnaryInfC{$(0, \emptyset) \ \ \  \Gamma, [\Diamond (p \vee q), p \vee q]_{\bullet 1}$} \RightLabel{\scriptsize{$\rulet$} } \UnaryInfC{$(0, \emptyset) \ \ \  \Gamma, [\Diamond (p \vee q)]_{\bullet 1}$} \AxiomC{$\Gamma, \Diamond \top$} \RightLabel{\scriptsize{$\ruledtrunk$} } \BinaryInfC{$(0, \emptyset) \ \ \   \phi, \overline{p}, \Diamond \Diamond (p \vee q)$}
\RightLabel{\scriptsize{$\vee$} } \UnaryInfC{$(0, \emptyset) \ \ \   \overline{p} \vee \Diamond \Diamond (p \vee q)$}
    \DisplayProof }
    \end{center}
$\calH=\phi, \overline{p}, \Diamond \Diamond (p \vee q), [\Diamond (p \vee q), p \vee q, p, q]_{\bullet1}, [\Diamond(p \vee q), p \vee q, p, q]_{\bullet \ruled}, [[p \vee q, p, q]]_{\ruled}$ in the left leaf is a saturated sequent with $\Diamond$-formulas, crown components, $t=1$, and $\Box \Diamond\calH_c = \{\Diamond (p \vee q)\} \subseteq \{\Diamond (p \vee q)\}=\Sigma_c$. Hence, by Step~\ref{step_saturated_insufficeintbase},
\begin{equation}
\label{eq:ex:topleft}
A_p(1, \{\Diamond (p \vee q)\} ; \calH)\quad=\quad\bullet: \overline{p} \,\spdisj\, \bullet 1:q \,\spdisj\, \bullet \ruled:q \,\spdisj\ \ruled:q.
\end{equation}
Applications of rule $\vee$ do not change the interpolant (Step~\ref{step_prop}, row 3). To compute  $A_p(0,\varnothing; \Gamma, [\Sigma_1]_{\bullet1})$ for the conclusion of $\doubledrule$, we  convert~\eqref{eq:ex:topleft} into an SDNF
\small\[
\Bigl(\bullet: \overline{p} \spconj \bigspconj_{\sigma \in \{\bullet 1, \bullet \ruled, \ruled\}} \sigma : \top\Bigr) \spdisj \bigspdisj_{\tau \in \{\bullet 1, \bullet \ruled, \ruled\}}\Bigl(\tau: q \spconj \bigspconj_{\sigma \in \{ \bullet, \bullet 1, \bullet \ruled, \ruled\} \setminus \{ \tau \}} \sigma : \top\Bigr).
\]
\normalsize Now, by Step \ref{step_saturated_doubled}, and converting into a new SDNF, we get $A_p(0, \emptyset ; \Gamma, [\Sigma_1]_{\bullet1})\equiv{}$
\begin{align*}   
\big(\bullet: (\overline{p} \wedge \Diamond \top) \spconj \bullet 1: (\top \wedge \Diamond \top)\big)
\spdisj
 \big(\bullet: (\top \wedge \Diamond \top)  \spconj \bullet 1: (q \wedge \Diamond \top)\big)
\spdisj 
\\
\big(\bullet: (\top \wedge \Diamond q)  \spconj \bullet 1: (\top \wedge \Diamond \top)\big) 
\spdisj
\big(\bullet: (\top \wedge \Diamond \top)  \spconj \bullet 1: (\top \wedge \Diamond q)\big)\rlap{.}\phantom{\spdisj}
\end{align*}
\normalsize Further applications of $\vee$ and $\rulet$ keep this interpolant intact. Note that the application of $\ruled'_t$ does not require to continue proof search for the right branch. Instead, Step~\ref{step_saturated_dtrunk} prescribes that $A_p(0, \emptyset ; \phi,  \overline{p}, \Diamond \Diamond (p \vee q))\equiv \bigl(\bullet : \overline{p} \spdisj \bullet : \Diamond \top\bigr) \spconj$ 
\begin{align*}   
\Bigl(\big(\bullet: (\overline{p} \wedge \Diamond \top \wedge \Diamond(\top \wedge \Diamond \top))\big)
\spdisj
 \big(\bullet: (\top \wedge \Diamond \top \wedge \Diamond (q \wedge \Diamond \top))\big)
\spdisj \phantom{\bullet : \Box \bot\Bigr)}
\\
\phantom{\Bigl(}\big(\bullet: (\top \wedge \Diamond q \wedge\Diamond(\top \wedge \Diamond \top))\big) 
\spdisj
\big(\bullet: (\top \wedge \Diamond \top \wedge\Diamond(\top \wedge \Diamond q))\big)\spdisj \bullet : \Box \bot\Bigr)\rlap{.} 
\end{align*}
\normalsize Simplifying, we 
finally 
obtain 
\begin{equation}
\label{eq:ex:finalint}
A_p(0, \emptyset ; \phi) \equiv\bullet : \Bigl((\overline{p} \vee \Diamond \top) \wedge \bigl((\overline{p} \wedge \Diamond \top) \vee \Diamond q \vee \Diamond \Diamond q \vee \Box \bot\bigr)   \Bigr) \equiv
\bullet : (\overline{p} \vee \Diamond \Diamond q) .
\end{equation}
To check that $\overline{p} \vee \Diamond \Diamond q$ is  a uniform Lyndon interpolant for $\phi$ w.r.t.~literal~$p$, it is sufficient to verify that~\eqref{eq:ex:finalint} is a BLU interpolant for $\calG$ by checking the conditions in Def.~\ref{def:BLUIP}. We only check BLUIP\eqref{BLUIP:3} as the least trivial. If $\calM,\calI \nvDash \bullet : (\overline{p} \vee \Diamond \Diamond q)$ for an interpretation $\calI$ into a $\Kfive$-model $\calM=(W,R,V)$, then, by Defs.~\ref{def:multiformula} and~\ref{Dfn:  interpretation},  $\calM, \rho \nvDash \overline{p} \vee \Diamond \Diamond q$ for the root $\rho$ of~$\calM$. For $\ell=p$, we have an $\ell$-bisimulation 
$(\calM',\calI) \leq_\ell (\calM,\calI)$ for $\calM'=(W,R,V')$ with $V'(p) = \{\rho\}$ and  $V'(r) = V(r)$ for $r\ne p$  since $\mathbf{literals}_p$ allows to turn $p$ from true to false. It is easy to see that $\calM',\rho \nvDash \overline{p} \vee \Diamond \Diamond (p \vee q)$. Thus, $\calM', \calI \nvDash \bullet : \phi$.
\end{example}

The following properties of the algorithm are proved in the Appendix.

\begin{restatable}{lemma}{Aptechnical}
\label{lem:Ap_properties}
All recursive calls $A_\ell(t,\Sigma_c; \calG)$ in a proof search tree of $A_\ell(0,\varnothing; \phi)$ have the following properties:
    \begin{compactenum}
        \item\label{thm_item:terminating}
        The algorithm is terminating.
        \item\label{thm_item:dtrunk}
        When Step~\ref{step_saturated_dtrunk} is applied, $t=0$ and every branch going up from it consists of Steps~\ref{step_prop}--\ref{step_modal} followed by either final Step~\ref{step_initial} or continuation via Step~\ref{step_saturated_doubled}.
        \item\label{thm_item:doubled}
        After Step~\ref{step_saturated_doubled} is applied, every branch going up from it consists of  Steps~\ref{step_prop} followed by a call  $A_\ell(1,\Box \Diamond \calG_c;\calG, [\Theta]_{\bullet\ruled}, [[\Phi]]_{\ruled})$ of one of the following types:
        \begin{compactenum}
            \item 
            sufficient and final when calculated via Step~\ref{step_initial};
            \item 
            sufficient and propositionally saturated when calculated via Step~\ref{step_modal}, with every branch going up from there consisting of more \mbox{Steps~\ref{step_prop}--\ref{step_modal}}, followed by either final Step~\ref{step_initial} or continuation via Step~\ref{step_saturated_doubled};\looseness=-1
            \item \label{case:hard}
            insufficient and saturated when calculated via Step~\ref{step_saturated_insufficeintbase}.
        \end{compactenum}
\end{compactenum}
\end{restatable}

\begin{restatable}{theorem}{BLUIPKfive}
\label{thm:BLUIP}
    Logic\/ $\Kfive$ has the BLUIP and, hence, the ULIP.
\end{restatable}
\begin{proof}
It is sufficient to prove  that, once the algorithm starts on $A_\ell(0,\varnothing; \phi)$, then every sufficient call $A_\ell(t,\Sigma_c;\calG)$ in the proof search returns a BLU interpolant for a $\Kfive$-sequent~$\calG$. Because the induction on the proof-search  is quite technical and involves multiple cases, we  demonstrate only a few representative cases, relegating most to the Appendix and omitting simple ones, e.g., BLUIP\eqref{BLUIP:1},  altogether.\looseness=-1

\begin{asparadesc}

\item[\textup{\textbf{BLUIP\eqref{BLUIP:2}}}] We show that $\calM, \calI \vDash A_\ell(t, \Sigma_c;\calG)$ implies $\calM, \calI \vDash \calG$ for any interpretation~$\calI$ of~$\calG$ into any $\Kfive$-model~$\calM=(W,R,V)$.  
    The hardest among  Steps~\mbox{\ref{step_initial}--\ref{step_modal}} is \textbf{Step~\ref{step_modal} using row~5} in Table~\ref{table:Ap}. Let $\calG=\calG', \Box \phi$ and $\calM, \calI \vDash A_\ell(t, \Sigma_c;\calG', \Box \phi)$ for
         \begin{equation}\label{Ap_box}
        A_\ell(t, \Sigma_c ;\quad \calG', \Box \phi) \quad=\quad \bigspconj_{i=1}^h \left(\bullet : \Box \delta_i \spdisj \bigspdisj_{\tau \in \calG} \tau : \gamma_{i,\tau} \right),
        \end{equation}
   i.e., for each $1\leq i \leq h$ either $\calM, \rho \vDash \Box \delta_i$ or $\calM,\calI(\tau) \vDash \gamma_{i,\tau}$ for some $\tau \in \calG$. 
   For an arbitrary $v$ such that $\rho R v$ and the the smallest $j$ such that~$\bullet j \notin \calG$, clearly $\calI_v = \calI \sqcup \{(\bullet j, v )\}$ is an interpretation of $\calG', \Box \varphi, [\varphi]_{\bullet j}$ into $\calM$. Since \mbox{$\calM, \calI_v(\bullet j)\vDash \delta_i$} whenever $\calM, \rho \vDash \Box \delta_i$, it follows that for each $1\leq i \leq h$ either $\calM, \calI_v(\bullet j) \vDash \delta_i$ or $\calM,\calI_v(\tau) \vDash \gamma_{i,\tau}$ for some $\tau \in \calG$, i.e., $\calM, \calI_v \vDash A_\ell\bigl(t, \Sigma_c ; \calG', \Box \varphi, [\varphi]_{\bullet j}\bigr)$ for
        \begin{equation}\label{Ap_box_premise}
        A_\ell\bigl(t, \Sigma_c ;\quad \calG', \Box \varphi, [\varphi]_{\bullet j}\bigr) \quad\equiv\quad \bigspconj_{i=1}^h \left(\bullet j : \delta_i  \spdisj \bigspdisj_{\tau \in \calG} \tau : \gamma_{i,\tau} \right).
        \end{equation}
    By IH, $\calM, \calI_v \vDash  \calG', \Box \varphi, [\varphi]_{\bullet j}$ whenever $\rho R v$. If $\calM, \rho \vDash \Box \phi$, then $\calM, \calI \vDash \calG$. Otherwise, $\calM, \calI_v(\bullet j) \nvDash \phi$ for some  $v$ with $\rho R v$. For it, $\calM, \calI_v \vDash \calG'$, hence, $\calM, \calI \vDash \calG$.
 
    The only other case we consider (here) is  \textbf{Step~\ref{step_saturated_doubled}}. 
    Let $\calM, \calI \vDash A_\ell(t, \Sigma_c; \calG)$ for $A_\ell(t, \Sigma_c; \calG)$ from~\eqref{Ap_saturated_doubled}, i.e., for some $1 \leq i \leq h$ we have $\calM, \rho \vDash \Diamond \delta_i$, and $\calM, \calI(\bullet 1)\vDash \Diamond \delta'_i$, and $\calM, \calI(\tau) \vDash \gamma_{i,\tau}$ for all $\tau \in \calG$. In particular, $\calM, v \vDash \delta_i$ for some $\rho Rv$  and $\calM, u \vDash \delta'_i$ for some $\calI(\bullet 1)Ru$. 
    Let $\calM'$ be obtained by copying $u$ into $u'$ away from the root in $\calM$ and let $\calJ = \calI \sqcup \{(\bullet \ruled,v), (\ruled, u')\}$ be a well-defined interpretation. 
    $\calM',\calJ \vDash A_\ell(1, \Box \Diamond \calG_c ; \calG, [\{\psi \mid \Diamond \psi \in \calG \}]_{\bullet \ruled}, [[ \{ \chi \mid \Diamond \chi \in \calGcrown \} ]]_{\ruled})$, as \eqref{Ap_saturated_SDNFdoubled} is true for $\calM'$ and $\calJ$. By IH, $\calM',\calJ \vDash  \calG, [\{\psi \mid \Diamond \psi \in \calG \}]_{\bullet \ruled}, [[ \{ \chi \mid \Diamond \chi \in \calGcrown \} ]]_{\ruled}$. 
    If $\calM', v \vDash\psi$ for some $\Diamond \psi \in \calG$ or $\calM', u' \vDash\chi$ for some $\Diamond \chi \in \calGcrown$, then $\calM', \calJ \vDash \calG$ because of $\Diamond \psi$ or $\Diamond \chi$ respectively. Otherwise, also $\calM',\calJ \vDash \calG$. Since we have $(\calM, \calI) \sim (\calM',\calJ)$ by Lemma~\ref{lem:copying},  we have $\calM, \calI \vDash \calG$ by Lemma~\ref{lem_modal_equivalence}\eqref{lem_modal_equivalence:full} in all cases.

\item[\textup{\textbf{BLUIP\eqref{BLUIP:3}}}] We show the following statement by induction restricted to  sufficient calls: if $\calM, \calI \nvDash A_\ell(t, \Sigma_c ; \calG)$, then $\calM', \calJ' \nvDash \calG$ for some interpretation~$\calJ'$  of $\calG$ into another $\Kfive$-model~$\calM'$ such that $(\calM',\calJ') \leqpos (\calM,\calI)$. Here we only consider \textbf{Step~\ref{step_drule}} as the other steps are sufficiently similar to~{\sf K}~and~{\sf S5} covered in~\cite{van2021,GieJalKuz21arx}.  Among the four subcases, Step~\ref{step_saturated_base} is tedious but conceptually transparent. Step~\ref{step_saturated_insufficeintbase} is trivial because the induction statement is only for sufficient calls while Step~\ref{step_saturated_insufficeintbase}  calls are insufficient   by Lemma~\ref{lem:Ap_properties}. Out of remaining two steps we only have space for \textbf{Step~\ref{step_saturated_doubled}}, which is conceptually the most interesting because its recursive call may be insufficient, precluding the use of IH for it. Let $\calM, \calI \nvDash A_\ell(t, \Sigma_c; \calG)$ for $A_\ell(t, \Sigma_c; \calG)$ from~\eqref{Ap_saturated_doubled}. 
 
    We first modify $\calM$ and $\calI$ to obtain an injective interpretation $\calI'$ into a $\Kfive$-model $\calN'=(W',R',V')$ such that $W' \setminus\mathit{Range}(\calI')$ is not empty and partitioned into pairs $(v,u)$ with       
            $\calI'(\bullet) R v$  and not $\calI'(\bullet) Ru$.
         To this end we employ copying as per Def.~\ref{def:copying}, constructing a sequence of interpretations $\calI_i$ from $\calG$ into models~$\calN_i=(W_i,R_i,V_i)$ starting from $\calN_0=\calM$ and $\calI_0 = \calI$ as follows:
    \begin{compactenum}
       \item \label{construction_model_injective_diamond}
        If $\calI_i(\tau_1)=\calI_i(\tau_2)$ for $\tau_1 \ne \tau_2$, obtain $\calN_{i+1}$ by copying $\calI_i(\tau_2)$ to a new world~$w$ and  redirect $\tau_2$ to this new world, i.e., $\calI_{i+1}= \calI_i \sqcup \{(\tau_2,w)\} \setminus \{(\tau_2,\calI_i(\tau_2))\}$.
        
     \item \label{step:construction_nonemptypairs}
        If $\calI_{K-1}$ is injective but $W_{K-1} \setminus \mathit{Range}(\calI_{K-1})=\varnothing$, obtain $\calN_{K}$ by copying $\calI_{K-1}(\bullet 1)$ to a new world $y$. Set $\calI_{K} =  \calI_{K-1}$. Now $W_{K} \setminus \mathit{Range}(\calI_{K})\ne\varnothing$.

    \item \label{step:construction_pairs}
        Finally, define the two sets  $Y = \{y \in W_{K} \setminus\mathit{Range}(\calI_{K}) 
        \mid \calI_{K}(\bullet) R_{K} y\}$ and $Z = \{z \in W_K \setminus\mathit{Range}(\calI_K) 
        \mid \text{not }\calI_K(\bullet) R_K z\}$ and
        obtain $\calN'$ by  copying:
        \begin{compactitem}
        \item for each $y \in Y$, copy $\calI_K(\bullet1)$ away from the root to a new world $y_2$;
        \item for each $z \in Z$, copy $\calI_K(\bullet1)$  to a new world $z_1$.
        \end{compactitem}
        Then $\calI' = \calI_K$ is an injective interpretation of $\calG$ into $\calN'$. 
\end{compactenum}
Note that  
        $ 
        W' \setminus\mathit{Range}(\calI')=Y \sqcup Z \sqcup \{y_2 \mid y \in Y\} \sqcup \{z_1 \mid z \in Z\}
        \ne \emptyset 
        $. Further,  $\calI'(\bullet) R' y$ for all $y \in Y$, and not $\calI'(\bullet) R' y_2$ for all $y \in Y$, and $\calI'(\bullet) R' z_1$ for all $z \in Z$, and not $\calI'(\bullet) R' z$ for all $z \in Z$. Thus, we obtain the requisite partition $P=\{(y,y_2)\mid y \in Y\}\sqcup\{(z_1,z)\mid z \in Z\}\ne \varnothing$  of the non-empty $W' \setminus\mathit{Range}(\calI')$.
        
        It is clear that $(\calN',\calI') \sim (\calM, \calI)$. So $\calN', \calI' \nvDash A_\ell(t, \Sigma_c ; \calG)$ by Lemma~\ref{lem_modal_equivalence}, i.e., for each $1 \leq i \leq h$ we have  $\calN', \rho \nvDash \Diamond \delta_i$ for $\rho = \calI'(\bullet)$, or $\calN', \calI'(\bullet 1)\nvDash \Diamond \delta'_i$, or 
        $\calN', \calI'(\tau) \nvDash \gamma_{i,\tau}$ for some $\tau \in \calG$. 
Thus, for any $(v,u)\in P$ and each  $1 \leq i \leq h$, we have  $\calN', v \nvDash \delta_i$, or $\calN', u\nvDash  \delta'_i$, or $\calN', \calI'(\tau) \nvDash \gamma_{i,\tau}$ for some $\tau \in \calG$. Hence,   \eqref{Ap_saturated_SDNFdoubled} is false under injective interpretation $\calJ_{v,u} = \calI' \sqcup \{(\bullet\ruled, v), (\ruled, u)\}$ into $\calN'$, i.e., abbreviating  $\Theta = \{\psi \mid \Diamond \psi \in \calG \}$ and $\Phi = \{ \chi \mid \Diamond \chi \in \calGcrown \}$, we get $\calN', \calJ_{v,u} \nvDash A_\ell(1, \Box \Diamond \calG_c; \calG, [\Theta]_{\bullet \ruled}, [[ \Phi ]]_{\ruled}).$

 Ordinarily, here we would use IH, but this is only possible for sufficient calls, which, alas,  is not guaranteed for~\eqref{Ap_saturated_SDNFdoubled}. What is known by Lemma~\ref{lem:Ap_properties}\eqref{thm_item:doubled} is that every branch going up from~\eqref{Ap_saturated_SDNFdoubled} leads to a call of the form 
  \begin{equation}\label{multiformulas}
    A_\ell(1, \Box \Diamond \calG_c ;\quad \calG, [\Theta_j]_{\bullet \ruled}, [[ \Phi_j ]]_{\ruled}),
   \end{equation}
    where~$\Theta_j \supseteq \Theta$ and~$\Phi_j \supseteq \Phi$, that returns multiformula $\mho_j$ and is either sufficient or insufficient but saturated. Let~$\Xi $ denote the multiset of these  multiformulas~$\mho_j$ returned by all these calls. Since Step~\ref{step_prop} is the only one used between that call and all the calls comprising~\eqref{multiformulas}, it is clear that  \eqref{Ap_saturated_SDNFdoubled} is  their conjunction, i.e.,
        $
        A_\ell(1, \Box \Diamond \calG_c; \calG, [\Theta]_{\bullet \ruled}, [[ \Phi ]]_{\ruled}) \equiv \bigspconj\nolimits_{\mho_j \in \Xi} \mho_j
    $.
    Collecting all this together, we conclude that for each pair $ (v,u) \in P$ there is some $\mho_{v,u} \in \Xi$ such that
    \begin{equation}
    \label{eq:mainthingy}
        \calN', \calJ_{v,u} \nvDash \mho_{v,u}.
    \end{equation}
    We distinguish between two cases.
    First,  suppose for at least one pair  $(v,u) \in P$ there is a sufficient $\mho_{v,u}= A_\ell(1, \Box \Diamond \calG_c ; \calG, [\Theta_{v,u}]_{\bullet \ruled}, [[ \Phi_{v,u} ]]_{\ruled})$ satisfying~\eqref{eq:mainthingy}. By~IH for this $\mho_{v,u}$ there is an interpretation~$\calJ_0'$  into a $\Kfive$-model~$\calM'$ such that $(\calM', \calJ_0') \leqpos (\calN', \calJ_{v,u})$ and~$\calM', \calJ_0' \nvDash \calG, [\Theta_{v,u}]_{\bullet \ruled}, [[ \Phi_{v,u} ]]_{\ruled}$. Thus, $\calM', \calJ' \nvDash \calG$
  for  $\calJ' = \calJ_0'\! \upharpoonright\! \Lab(\calG)$. Finally, by restricting to labels of $\calG$, we can see that  \begin{equation}
      \label{eq:bisimchain}
  (\calM',\calJ')\quad\leqpos\quad (\calN', \calI')\quad \sim\quad (\calM, \calI).
  \end{equation}

Otherwise,  \eqref{eq:mainthingy} does not hold for any pair $(v,u)\in P $ and any sufficient \mbox{$\mho_{v,u} \in \Xi$}. In this case,  $\calN', \calJ_{v,u} \nvDash \bigspconj\nolimits_{\mho_j \in \Xi} \mho_j$ guarantees the existence of an insufficient $\mho_{v,u}\in \Xi$ for each pair $(v,u) \in P$ such that \eqref{eq:mainthingy} holds.
Since all these $\mho_{v,u}$ are insufficient, we cannot use IH. Instead, we construct $\calM'$ and $\calJ'$ directly by changing $\ell$ from true to false if needed based on $\calG$ within $\mathit{Range}(\calI')$  and based on $\mho_{v,u}$'s outside of this range. Thanks to $\calI'$ being injective, we do not need to worry about conflicting requirements from different components of~$\calG$. Similarly, $P$ being a partition prevents conflicts outside $\mathit{Range}(\calI')$. Let $\calM' = (W',R',U')$ be~$\calN'$ with  $V'$ changed into  $U'$. 
    We define $\mfal{V'}{T}{\ell}$ as the valuation that makes $\ell$ false in all worlds from $T \subseteq W'$, i.e.,  $(\mfal{V'}{T}{\ell}) (q)= V'(q)$ for all $q\notin\{\ell,\overline{\ell}\}$, while 
    \[
    (\mfal{V'}{T}{\ell})(p) = 
    \begin{cases}
    V'(p) \setminus T & \text{if $\ell =p$},
    \\
    V'(p) \cup T & \text{if $\ell=\overline{p}$}
    \end{cases}
    \]
    for $p \in \{\ell,\overline{\ell}\}$.
    Using this notation, we define $U' = \mfal{V'}{T_\calG}{\ell}$ where
    \begin{multline}
    \label{valuation_calG}
        T_\calG =  \{ \calI'(\sigma) \mid \sigma : \ell \in \calG \} 
            \sqcup
            \{ v \mid (v,u) \in P\text{ and }\bullet\ruled: \ell \in \mho_{v,u}  \} 
            \sqcup \\
            \{u \mid (v,u) \in P \text{ and } \ruled:\ell \in \mho_{v,u}\}.
    \end{multline} 
    Finally, $\calJ' = \calI'$. It is clear that \eqref{eq:bisimchain}~holds for these $\calM'$ and $\calJ'$.

It remains to show that $\calM',\calJ' \nvDash \calG$. This is done by mutual induction on the construction of formula $\phi$ for the following three induction statements
\begin{align}
\label{calG_property}
    \sigma : \phi \in \calG  &\Longrightarrow \calM', \calI'(\sigma) \nvDash \phi,\\
\label{ThetA^+_property}
    \bullet\ruled :\phi \in \mho_{v,u} &\Longrightarrow \calM', v \nvDash \phi,\\
\label{Phi_property}
    \ruled: \phi \in \mho_{v,u} &\Longrightarrow  \calM', u \nvDash \phi.
\end{align}

\begin{asparadesc}
\item[Case $\phi = \ell'\in \Lit \setminus\{\ell, \overline{\ell}\}$.]   By Lemma~\ref{lem:Ap_properties}\eqref{thm_item:doubled}, all $\mho_{v,u}$ are computed by Step~\ref{step_saturated_insufficeintbase} due to their insufficiency, i.e., $\mho_{v,u}=\LitDis_\ell(\calG, [\Theta_{v,u}]_{\bullet \ruled}, [[ \Phi_{v,u} ]]_{\ruled})$. \eqref{ThetA^+_property} and \eqref{Phi_property} follow from \eqref{eq:mainthingy} and \eqref{Ap_saturated_base} because $\calM'$ agrees with $\calN'$ on $\ell' \notin\{\ell,\overline{\ell}\}$. Similarly, since  $\calJ_{v,u}$ agrees with  $\calJ'=\calI'$ on $\Lab(\calG)$, \eqref{calG_property} follows by using $\mho_{v,u}$ for any $(v,u) \in P \ne \varnothing$. 

\item[Case $\phi = \overline{\ell}$]       
        is analogous to the previous one. The only difference is the reason why $\calM'$ agrees with $\calN'$ on $\overline{\ell}$. Here,  $\sigma : \overline{\ell}\in\calG$ implies $\sigma : \ell\notin\calG$ because $\calG$ was processed by Step~\ref{step_saturated_doubled} not Step~\ref{step_initial}.
        Therefore, $\calI'(\sigma) \notin T_\calG$ by the injectivity of~$\calI'$, and  $\overline{\ell}$ was not made true in $\calI'(\sigma)$, ensuring \eqref{calG_property}. The argument for \eqref{ThetA^+_property} and \eqref{Phi_property} is similar, except $\bullet\ruled/\ruled : \overline{\ell}$ is taken from  $\mho_{v,u}$ processed by Step~\ref{step_saturated_insufficeintbase} not Step~\ref{step_initial}.
 
\item[Case $\phi = \ell$.] All of \eqref{calG_property}--\eqref{Phi_property}  follow from \eqref{valuation_calG}.

\item[Cases $\phi = \phi_1 \wedge \phi_2$ and $\phi = \phi_1 \vee \phi_2$] are standard and follow by IH due to saturation of $\calG$ for~\eqref{calG_property}  and $\mho_{v,u}$ for \eqref{ThetA^+_property} and \eqref{Phi_property}.

         \item[Case $\phi=\Box \xi$.]
       If~$\sigma : \Box \xi \in \calG$, then by saturation of $\calG$, there is a~$\tau$ such that $\tau : \xi \in \calG$ and $\calI'(\sigma)R'\calI'(\tau)$: if $\sigma = \bullet$, then $\tau = \bullet j$ for some~$j$, while if $\sigma \ne \bullet$, then $\tau \ne \bullet$. By IH\eqref{calG_property}, $\calM',\calI'(\tau) \nvDash \xi$, and  $\calM', \calI'(\sigma) \nvDash \Box \xi$. 
       
       If 
        $\bullet \ruled/\ruled : \Box \xi \in \mho_{v,u}$, then $\Box \xi \in \Box \Diamond \calG_c$ by conditions of Step~\ref{step_saturated_insufficeintbase} due to \eqref{multiformulas}, i.e., 
        $\Box \xi \in \calG_c$. By saturation of~$\calG$, there is a~$\tau \ne \bullet$ such that~$\tau : \xi \in \calG$. Since $v$, $u$, and $\calI'(\tau)$ are all in the cluster $C$ of $\calM'$, we have $vR'\calI'(\tau)$ and $uR'\calI'(\tau)$. It remains to use IH\eqref{ThetA^+_property} and IH\eqref{Phi_property}.

               \item[Case $\phi=\Diamond \xi$.]
         First consider $\sigma = \bullet$ and $\bullet : \Diamond \xi \in \calG$. Since $\calI'(\bullet)=\rho$ is the root,  $\rho R' w$ implies  either $w = \calI'(\bullet j)$ for some $j$ or $w \notin \mathit{Range}(\calI')$. In the former case, $\bullet j : \xi \in \calG$
         by saturation of~$\calG$,  so~$\calM', w \nvDash \xi$ by IH\eqref{calG_property}. In the latter case,
        $(w,u) \in P$ for some $u$. Recall for~$A_\ell(1, \Box \Diamond \calG_c ; \calG, [\Theta_{w,u}]_{\bullet \ruled}, [[ \Phi_{w,u} ]]_{\ruled})$ that we have $\Theta_{w,u} \supseteq \Theta = \{ \psi \mid \Diamond \psi \in \calG\} \ni \xi$. Hence, $\bullet\ruled: \xi \in \mho_{w,u}$ and $\calM', w \nvDash \xi$ by IH\eqref{ThetA^+_property}. Since $\calM', w \nvDash \xi$ for all $\calI'(\bullet)=\rho R' w$, we conclude $\calM', \calI'(\bullet) \nvDash \Diamond \xi$.

        If $\sigma \ne \bullet$ and $\sigma : \Diamond \xi \in \calG$, the argument is similar. But additionally we may have $w = \calI'(k)$ for some $k$ or $(v,w) \in P$ for some $v$. In the former case, $k : \xi \in \calG$
         by saturation of~$\calG$,  so~$\calM', w \nvDash \xi$ by IH\eqref{calG_property}. In the latter case,  $\Phi_{v,w} \supseteq \Phi = \{ \chi \mid \Diamond \chi \in \calG_c\} \ni \xi$. Hence, $\ruled: \xi \in \mho_{v,w}$ and $\calM', w \nvDash \xi$ by IH\eqref{Phi_property}. Since $\calM', w \nvDash \xi$ for all $\calI'(\sigma) R' w$, we conclude $\calM', \calI'(\sigma) \nvDash \Diamond \xi$.
        
        If $\bullet\ruled/\ruled : \Diamond \xi \in \mho_{v,u}$, then, similar to  the analogous subcase of $\Box \xi$, conditions of Step~\ref{step_saturated_insufficeintbase} imply that $\Diamond \xi \in \calG_c$, i.e, $\tau_0 : \Diamond \xi\in \calG$ for some $\tau_0 \ne \bullet$. Then $\tau : \xi \in \calG$ for all $\tau \ne \bullet$ by saturation of $\calG$. Thus, $\calM',\calI'(\tau)\nvDash \xi$ for all $\tau \ne \bullet$ by IH\eqref{calG_property}. For each $y\notin\mathit{Range}(\calI')$ such that $\rho R' y$, there is $x$ such that  $(y,x) \in P$ and $\bullet \ruled : \xi \in \mho_{y,x}$ because $\Theta_{y,x} \supseteq \Theta\ni \xi$. Hence, $\calM', y \nvDash \xi$ by IH\eqref{ThetA^+_property}. Finally, for each $x\notin\mathit{Range}(\calI')$ such that not  $\rho R' x$, there is $y$ such that $(y,x) \in P$ and $ \ruled : \xi \in \mho_{y,x}$ because $\Phi_{y,x} \supseteq \Phi\ni \xi$. Hence, $\calM', x \nvDash \xi$ by IH\eqref{Phi_property}.  We have shown that  $\calM', w \nvDash \xi$ whenever $v R' w$ ($u R' w$). Thus, $\calM', v \nvDash \Diamond \xi$ and $\calM', u \nvDash \Diamond \xi$. \qed
        \end{asparadesc}
\end{asparadesc}
\end{proof}

\section{Conclusion}
\label{sect:concl}
We presented layered sequent calculi for several extensions of modal logic $\Kfive$: namely, $\Kfive$ itself, $\KDfive$, $\Kfourfive$, $\KDfourfive$, $\KBfive$,~and~$\Sfive$. By leveraging the simplicity of Kripke models for these logics, we were able to formulate these calculi in a modular way and obtain optimal complexity upper bounds for proof search. We used the calculus for $\Kfive$ to 
obtain the first syntactic (and, hence, constructive) proof of the uniform Lyndon interpolation property for $\Kfive$.

Due to the proof being technically involved, space considerations prevented us from extending the syntactic proof of  ULIP to $\KDfive$, $\Kfourfive$, $\KDfourfive$, $\KBfive$,~and~$\Sfive$. For $\Sfive$, layered sequents coincide with hypersequents, and we plan to upgrade the hypersequent-based syntactic proof of UIP from~\cite{vdGiessen22PhD} to ULIP (see also~\cite{GieJalKuz21arx}). As for $\KDfive$, $\Kfourfive$, $\KDfourfive$, and $\KBfive$, the idea is to modify  the method presented here for $\Kfive$  by using the layered sequent calculus for the respective logic and making other necessary modifications, e.g., to rule~$\doubledrule$, to fit the specific structure of the layers. We conjecture that the proof for $\Kfourfive$, $\KDfourfive$, and $\KBfive$ would be similar to that for $\Sfive$, whereas $\KDfive$ would more closely resemble $\Kfive$.

\paragraph{Acknowledgments.} Iris van der Giessen and Raheleh Jalali are grateful for the productive and exciting four-week research visit to the Embedded Computing Systems Group  at TU Wien. The authors thank the anonymous reviewers for their useful comments.

\bibliographystyle{splncs04}
\providecommand{\noopsort}[1]{}\providecommand{\giessen}[1]{}

\newpage
\appendix
\section{Appendix}

In this Appendix we  provide and elaborate on the proofs that could not be included into the main part because of space considerations.

\soundcomplete*
\begin{proof}
We prove a cycle of implications: from the left statement to the middle one, from the middle one to the right one, and, finally, from the right to the left one.
Most of the proof can be found in the main text. Here we elaborate only on the left-to-middle implication and the argument why the constructed countermodel is an $\logic$-model due to $\calG'$ being saturated in the right-to-left implication.

For the left-to-middle implication, i.e.,    $\vdash_\layL \calG \quad\Longrightarrow\quad \vDash_\logic \iota(\calG)$, which is proved by induction on an $\layL$-derivation of $\calG$, we  show the case of  $\Diamond_c$, a rule of all considered calculi. Of the many flavors of~$\Diamond_c$, depending on what brackets~$\llbracket \enspace \rrbracket$~and~$ \llparenthesis \enspace \rrparenthesis$ represent, we show the following application for $\Kfive$ and $\KDfive$:
    \begin{center}
        $\vlinf{\Diamond_c}{}{\mathcal{G}, [[ \Sigma, \Diamond \varphi ]], [ \Pi ]}{\mathcal{G}, [[ \Sigma, \Diamond \varphi ]], [ \Pi, \phi ]}$
    \end{center}
    Let $\calM$ be an $\logic$-model with root $\rho$ with $\calM, \rho \vDash \iota(\calG) \vee \Box \Box (\bigvee \Sigma \vee \Diamond \phi) \vee \Box (\bigvee \Pi \vee \phi)$. If $\calM, \rho \vDash \iota(\calG)$, $\calM, \rho \vDash \Box \Box(\bigvee \Sigma \vee \Diamond \phi)$, or $\calM, \rho \vDash \Box (\bigvee \Pi)$, we are done. If not, there is a $v$ such that $\rho R v$ and $\calM, v \vDash \phi$. By the form of the models, indeed, $\calM, \rho \vDash \Box \Box \Diamond \phi$, hence $\calM, \rho \vDash \Box \Box (\bigvee \Sigma \vee \Diamond \phi)$. Therefore, we have \mbox{$\vDash_\logic \iota(\calG) \vee \Box \Box (\bigvee \Sigma \vee \Diamond \phi) \vee \Box(\bigvee \Pi)$}.

   To prove that the countermodel $\calM$ constructed for the contraposition of the right-to-left implication is an $\logic$-model, we use the saturation of $\calG'$. For example, in case of $\KDfive$, the root has a successor because of $\ruled_t$-saturation. By induction on~$\phi$, for each $\sigma : \phi \in \calG'$ we have $\calM, \sigma \nvDash \phi$, using the saturation of $\calG'$. Taking $\calI$ of $\calG$ into $\calM$ as the identity function (or $\calI(\bullet) = 1$ in case of $\KBfive$), we have $\calM, \calI \nvDash \calG$ as desired. \qed
\end{proof}

\modaleq*
\begin{proof}
    From $\calM, \calI \vDash \calG$, it follows that $\calM, \calI(\sigma) \vDash \phi$ for some $\sigma : \phi \in \calG$. Let us examine 1, where $Z$ is the bisimulation underlying $\leq_\ell$.  As $\calI(\sigma) Z  \calI'(\sigma)$, it is sufficient to show by induction on $\phi$, that $\calM, w \vDash \phi$ implies $\calM', w' \vDash \phi$ for every $wZw'$. We treat some cases. If $\phi = \overline{\ell}$, i.e., $\calM, w \nvDash \ell$, it follows from part~b) of \textbf{\textup{literals}}$_{\ell}$ (Def.~\ref{def: Bisimilarity}) that $\calM',w' \nvDash \ell$, hence $\calM', w' \vDash \overline{\ell}$. If $\phi = \Diamond \psi$, there is a $v$ such that $wRv$ and $\calM,v \vDash \psi$. By \textbf{\textup{forth}}, there exists~$v'$ such that $w'Rv'$ and $vZv'$. So by IH, $\calM', v' \vDash \psi$ and therefore $\calM', w' \vDash \Diamond \psi$. For $\phi = \Box \psi$, let $w'Rv'$. By  \textbf{\textup{back}}, there exists $v$ such that $vZv'$ and $wRv$. So, $\calM, v \vDash \psi$ and by IH we have $\calM',v' \vDash \psi$ as desired.\looseness=-1\qed
\end{proof}

\Aptechnical*

\begin{proof}  For~\eqref{thm_item:terminating}, note that Steps~\ref{step_initial}--\ref{step_modal} are terminating by termination of the rules modulo saturation by Theorem~\ref{thm:termination_proof_search}. In addition, the algorithm stops in Step~\ref{step_saturated_base} or~\ref{step_saturated_insufficeintbase} by definition. When $A_\ell(t, \Sigma_c ; \calG)$ is calculated by Step~\ref{step_saturated_dtrunk}, $\calG$ must consist only of the trunk, and after that the algorithm continues on layered sequents that contain the crown, meaning we will never encounter Step~\ref{step_saturated_dtrunk} again. 
    
    In particular, to prove property~\eqref{thm_item:dtrunk}, note that each branch going up from Step~\ref{step_saturated_dtrunk} will either stop in Step~\ref{step_initial}, or will be calculated in Step~\ref{step_drule}. In the latter case, note that Steps~\ref{step_saturated_base}--\ref{step_saturated_insufficeintbase} do not apply resulting in Step~\ref{step_saturated_doubled}. In particular,~$t$ is set to~$0$ in Step~\ref{step_saturated_dtrunk} and will not change in any Steps~\ref{step_prop}--\ref{step_modal} which excludes Step~\ref{step_saturated_insufficeintbase}.
        
    Finally, continuing proof of~\eqref{thm_item:terminating}, there is a bound on the number of applications of Step~\ref{step_saturated_doubled}. Indeed, for each branch of the algorithm the annotations $\Sigma_c \subseteq \calG_c$ of the multiformulas $A_\ell(t, \Sigma_c ; \calG)$ form a non-decreasing sequence of sets. These sets grow in Step~\ref{step_saturated_doubled}, because there $\Box \Diamond \calG_c \not \subseteq \Sigma_c$ (otherwise Step~\ref{step_saturated_insufficeintbase} would apply). Also, all formulas in these annotations are subformulas of the layered sequent at the root. So, if a branch of the algorithm did not stop earlier by Step~\ref{step_initial}, one would reach, after finitely many applications of Step~\ref{step_saturated_doubled}, a multiformula $A_\ell(t, \Sigma_c ; \calG)$ for which $\Box \Diamond \calG_c  \subseteq \Sigma_c$ that stops the algorithm in Step~\ref{step_saturated_insufficeintbase}.
    
Property~\ref{thm_item:doubled} says something more refined than we just observed. The form of $A_\ell(1,\Box \Diamond \calG_c;\calG, [\Theta], [[\Phi]])$ is justified by the rules in Steps~\ref{step_saturated_doubled}, \ref{step_initial}, and~\ref{step_prop}. If it is calculated by Step~\ref{step_initial}, then this leaf is sufficient by definition. Otherwise, applications of Steps~\ref{step_initial} and~\ref{step_prop} propositionally saturate the layered sequent. There are two cases. Either, $\Box \Diamond (\calG, [\Theta], [[\Phi]])_c \not \subseteq \Box \Diamond \calG_c$. Then any saturation of it will contain new modal formulas in the crown and will trigger Step~\ref{step_saturated_doubled}, if not stopped earlier in Step~\ref{step_initial}. Hence, in this case, $A_\ell(1,\Box \Diamond \calG_c;\calG, [\Theta], [[\Phi]])$ will be sufficient. Or, $\Box \Diamond (\calG, [\Theta], [[\Phi]])_c \subseteq \Box \Diamond \calG_c$. We will now show that it will not be saturated only propositionally, but in fact fully saturated. Each $\sigma : \Box \phi \in \calG$ is saturated because $\calG$~is. If $\Box \phi \in \Theta$ or $\Box \phi \in \Phi$, then also $\Box \phi \in \calG_c$ by assumption and so in both cases these occurrences are saturated by the saturation of $\calG$. For  $\sigma : \Diamond \phi \in \calG$, note that it is saturated for each corresponding component in $\calG$ by saturation and $\Theta \supseteq \{ \psi \mid \Diamond \psi \in \calG \}$ and $\Phi \supseteq \{\chi \mid \Diamond \chi \in \calG_c \}$, making it also saturated w.r.t.~components $\Theta$ and $\Phi$ (the latter is only needed if $\Diamond \phi \in \calG_c$). For $\Diamond \phi \in \Theta$ and $\Diamond \phi \in \Phi$ the same argument applies as $\Diamond \phi \in \calG_c$ by assumption. Now, as $A_\ell(1,\Box \Diamond \calG_c;\calG, [\Theta], [[\Phi]])$ is saturated and $\Box \Diamond (\calG, [\Theta], [[\Phi]])_c \subseteq \Box \Diamond \calG_c$, we apply Step~\ref{step_saturated_insufficeintbase} to it which makes it insufficient. \qed
\end{proof}

\BLUIPKfive*

\begin{proof}
In addition to the cases provided in the proof in the main text, here we show additional representative cases:

\subsubsection*{BLUIP\eqref{BLUIP:2}, some of the cases from Steps~\ref{step_initial}--\ref{step_modal}.}
\begin{compactdesc}
\item[row 2 of Table \ref{table:Ap}:]

        $\calG=\calG'$$\vlfill{q,\overline{q}}_\sigma$. Then either
        $\calM, \calI(\sigma) \vDash q$ or~$\calM, \calI(\sigma) \vDash \overline{q}$, hence~$\calM, \calI \vDash \calG$.
       
    \item[row 4 of Table \ref{table:Ap}:]
    
If $\calG=\calG'\vlfill{\varphi \land \psi}_\sigma$, then
    \begin{align}\label{Ap_conj}
        A_\ell(t, \Sigma_c; \calG) = {A_\ell\bigl(t, \Sigma_c ; \cal  G'\vlfill{\varphi \land \psi, \varphi}_\sigma \bigr)}  \spconj {A_\ell\bigl(t, \Sigma_c ; \calG'\vlfill{\varphi \land \psi, \psi}_\sigma \bigr)}.
        \end{align}
        Since  $\calM, \calI \vDash {A_\ell\bigl(t, \Sigma_c ; \calG'\vlfill{\varphi \land \psi, \varphi}\bigr)}$ and $\calM, \calI \vDash {A_\ell\bigl(t, \Sigma_c ; \calG'\vlfill{\varphi \land \psi, \psi}\bigr)}$, by~IH, $\calM, \calI \vDash \cal  G'\vlfill{\varphi \land \psi, \varphi}$ and~$\calM, \calI \vDash \calG'\vlfill{\varphi \land \psi, \psi}$. If $\calM, \calI(\sigma) \vDash \phi \land \psi$ we are done. If not,~$\calM, \calI(\sigma) \nvDash \phi$ or~$\calM, \calI(\sigma) \nvDash \psi$. Therefore,~$\calM, \calI \vDash \calG$.
 \item[row 9 of Table \ref{table:Ap}:]
 
        Suppose~$\calG=\calG', \llbracket \Sigma, \Diamond \phi \rrbracket_\sigma, \llparenthesis \Pi \rrparenthesis_\tau$. We have
        \[
        A_\ell(t, \Sigma_c ; \calG) = A_\ell\bigl(t, \Sigma_c ; \calG', \llbracket \Sigma, \Diamond \phi \rrbracket_\sigma, \llparenthesis \Pi , \phi \rrparenthesis_\tau \bigr)
        \]
        By IH,~$\calM, \calI \vDash \calG', \llbracket \Sigma, \Diamond \phi \rrbracket_\sigma, \llparenthesis \Pi , \phi \rrparenthesis_\tau$. Note that~$\calI(\sigma)R\calI(\tau)$ meaning that $\calM, \calI(\tau) \vDash \phi$ implies~$\calM, \calI(\sigma) \vDash \Diamond \phi$. Hence~$\calM, \calI \vDash \calG$.
\end{compactdesc}

\subsubsection*{BLUIP\eqref{BLUIP:2}, remaining cases from Step \ref{step_drule}.}
\begin{compactdesc}
\item[Steps~\ref{step_saturated_base} and~\ref{step_saturated_insufficeintbase}.]
        We can treat them  simultaneously because both are calculated according to~\eqref{Ap_saturated_base}. This case is easy as the truth of the interpolant implies that some literal in~$\calG$ is true.

        \item[Step~\ref{step_saturated_dtrunk}.]       Suppose~$\calM, \calI \vDash A_\ell(0, \Sigma_c; \Gamma)$  for the multiformula presented in~\eqref{Ap_saturated_dtrunk}
        \begin{equation}\label{BLUIP:ii_saturated_dtrunk_left}
        A_\ell(0, \Sigma_c; \Gamma) =  
        \Bigl(\bullet : \Box \bot \spdisj  
            \bigspdisj_{i=1}^{h} \big( \bullet : \Diamond \delta_i  \spconj  \bullet : \gamma_{i} \big) 
            \Bigr) 
            \spconj 
            \bigl(\bullet : \Diamond \top \spdisj \LitDis_\ell(\Gamma)  \bigr).
        \end{equation}
         Obviously, either~$\calM, \rho \nvDash \Box \bot$ or~$\calM, \rho \nvDash \Diamond \top$. In the latter case, it follows from the right conjunct of~\eqref{BLUIP:ii_saturated_dtrunk_left} that some literal in~$\Gamma$ is true, resulting in~$\calM, \calI \vDash \Gamma$. In the former case,
         based on the left conjunct of~\eqref{BLUIP:ii_saturated_dtrunk_left}, there is an $1 \leq i \leq h$ such that $\calM, \rho \vDash \Diamond \delta_i \wedge \gamma_i$. In particular, $\calM, v \vDash \delta_i$ for some  $v$ such that~$\rho R v$.  It follows that \[
         \calM, \calJ \vDash \bigspdisj_{i=1}^{h} \Bigl( \bullet 1
            : \delta_i \spconj \bullet : \gamma_{i} \Bigr)
            \] for $\calJ = \calI \sqcup \{(\bullet 1, v)\}$, i.e., $\calM, \calJ \vDash A_\ell(0, \Sigma_c ; \Gamma, [\{\psi \mid \Diamond \psi \in \Gamma \}]_{\bullet 1})$. By IH, $\calM, \calJ \vDash\Gamma, [\{\psi \mid \Diamond \psi \in \Gamma \}]_{\bullet 1}$. If $\calM, v \vDash \psi$ for some $\Diamond \psi \in \Gamma$, then $\calM, \rho \vDash \Diamond \psi$, and $\calM, \calI \vDash \Gamma$. Otherwise, $\calM, \calJ \vDash \Gamma$, so again $\calM, \calI\vDash \Gamma$.
\end{compactdesc}

\subsubsection*{BLUIP\eqref{BLUIP:3}, some of the cases from Steps~\ref{step_initial}--\ref{step_modal}.}
Let us start with the induction steps in which~$\calG$ is not saturated and 
$A_\ell(t, \Sigma_c ; \calG)$ is calculated according to a row in Table~\ref{table:Ap} (one of Steps~\ref{step_initial}--\ref{step_modal} in the algorithm):
\begin{compactdesc}
 \item[row 2 of Table \ref{table:Ap}:]
   $\calG$ cannot be of the form~$\calG'\vlfill{q,\overline{q}}_\sigma$ since $A_\ell(t, \Sigma_c ; \calG) = \sigma : \top$  and we assumed~$\calM, \calI \nvDash A_\ell(t, \Sigma_c ; \calG)$.

    \item[row 4 of Table \ref{table:Ap}:]
    
    If~$\calG=\calG'\vlfill{\varphi \land \psi}_\sigma$, then~$A_\ell(t, \Sigma_c; \calG)$ is defined in~\eqref{Ap_conj}. So either~$\calM, \calI \nvDash A_\ell\bigl(t, \Sigma_c ; \cal  G'\vlfill{\varphi \land \psi, \varphi}_\sigma \bigr)$ or~$\calM, \calI \nvDash  A_\ell\bigl(t, \Sigma_c ; \calG'\vlfill{\varphi \land \psi, \psi}_\sigma \bigr)$. The cases are symmetric, so let us only treat the former. By assumption, $\calG$~is a sufficient layered sequent and thus so is~$\calG'\vlfill{\varphi \land \psi, \varphi}_\sigma$ by definition. So we can apply the IH and get a model~$\calM'$ and interpretation~$\calI'$, such that~$(\calM', \calI') \leqpos (\calM, \calI)$ and~$\calM', \calI' \nvDash \calG'\vlfill{\varphi \land \psi, \varphi}_\sigma$. So, in particular,~$\calM', \calI' \nvDash \calG$.
    
    \item[row 5 of Table~\ref{table:Ap}:] If~$\calG=\calG', \Box \phi$, then~$A_\ell(t, \Sigma_c ; \calG)$ is defined in~\eqref{Ap_box} using the SCNF of~$A_\ell\bigl(t, \Sigma_c ; \calG', \Box \varphi, [\varphi]_{\bullet j}\bigr)$ as defined in~\eqref{Ap_box_premise}. From~$\calM, \calI \nvDash A_\ell(t, \Sigma_c ; \calG)$ it follows that for some~$1 \leq i \leq h$, both~$\calM, \calI(\bullet) \nvDash \Box \delta_i$ and also for all $\tau \in \calG$,~$\calM, \calI(\tau) \nvDash \gamma_{i,\tau}$. From the former it follows that there is~$v$ such that~$\calI(\bullet) R v$ and~$\calM, v \nvDash \delta_i$. Define a new interpretation $\calJ = \calI \sqcup \{(\bullet j, v) \}$. By inspection of~\eqref{Ap_box_premise} we can easily see that $\calM, \calJ \nvDash A_\ell\bigl(t, \Sigma_c ; \calG', \Box \varphi, [\varphi]_{\bullet j}\bigr)$. Note that~$\calG', \Box \varphi, [\varphi]_{\bullet j}$ is a sufficient layered sequent, and by IH  there are~$\calM'$ and~$\calJ'$ such that~$(\calM', \calJ') \leqpos (\calM, \calJ)$ and~$\calM', \calJ' \nvDash \calG', \Box \varphi, [\varphi]_{\bullet j}$. Now define $\calI' = \calJ' \upharpoonright Dom(\calI)$. Clearly,~$(\calM',\calI') \leqpos (\calM,\calI)$ and~$\calM, \calI \nvDash \calG$.

    \item[row 9 of Table \ref{table:Ap}:]

    If~$\calG=\calG', \llbracket \Sigma, \Diamond \phi \rrbracket_\sigma, \llparenthesis \Pi \rrparenthesis_\tau$, then~$A_\ell(t, \Sigma_c ; \calG)$ is obtained from a sufficient layered sequent as~$A_\ell\bigl(t, \Sigma_c ; \calG', \llbracket \Sigma, \Diamond \phi \rrbracket_\sigma, \llparenthesis \Pi , \phi \rrparenthesis_\tau \bigr)$. By IH, there is~$(\calM', \calI') \leqpos (\calM, \calI)$ such that $\calM',\calI' \nvDash \calG', \llbracket \Sigma, \Diamond \phi \rrbracket_\sigma, \llparenthesis \Pi , \phi \rrparenthesis_\tau$. In particular,~$\calM',\calI' \nvDash \calG$. 
\end{compactdesc}

\subsubsection*{BLUIP\eqref{BLUIP:3},  Step~\ref{step_saturated_base}.}
     Here $\calG$ has no diamond formulas and~$A_\ell(t, \Sigma_c;\calG)$ is calculated according to~\eqref{Ap_saturated_base}. Since~$\calM,\calI \nvDash A_\ell(t, \Sigma_c;\calG)$, we have by inspection of~\eqref{Ap_saturated_base} that for all $\sigma : \ell' \in \calG $~where $\ell' \in \Lit \setminus \{ \ell \}$,
        \begin{align}\label{Ap_base_literal}
        \calM, \calI(\sigma) \nvDash  \ell'.
        \end{align}
    We construct an $\ell$-bisimilar model~$\calM'$ with  interpretation~$\calI'$ such that  each occurrence of~$\ell$ in~$\calG$ is still falsified in the corresponding world in~$\calM'$ according to~$\calI'$. We proceed in two steps. Informally speaking, we first copy worlds to make an injective interpretation~$\calI'$ and after that we modify the valuation of~$p\in\{\ell,\overline{\ell}\}$, as desired.
    \begin{compactenum}
        \item\label{construction_model_injective}
        Divide the domain of~$\calI$ into equivalence classes such that~$\sigma$ and~$\tau$ belong to the same class if and only if~$\calI(\sigma) = \calI(\tau)$. One could say that each class is represented by a world in~$\calM$. For each equivalence class with~$n > 1$ elements, make~$n-1$ copies of the corresponding world and call this model $\calN = (W', R', V_\calN)$. Define~$\calI'$ to be the  interpretation assigning different labels in an equivalence class to copies of the worlds in such a way that~$\calI'$ is injective and, by Lemma~\ref{lem:copying},~$(\calN, \calI') \sim (\calM, \calI)$.

        \item\label{construction_model_valuation}
    Define~$\calM' = (W', R', V')$ to be the same as~$\calN$ except for valuations of~$p$:
        \[
         V'(p) = \begin{cases}
            V_\calN(p) \setminus \{ w \mid \sigma : p \in \calG, w = \calI(\sigma)\} & \text{ if } \ell = p.\\
            V_\calN(p) \cup \{ w \mid \sigma : \overline{p} \in \calG, w = \calI(\sigma)\} & \text{ if } \ell = \overline{p}.
        \end{cases}
        \]
        Indeed, by Definition \ref{def: Bisimilarity},~$(\calM',\calI') \leqpos (\calM, \calI)$ as desired. 
    \end{compactenum}
    Now we check that the constructed model~$\calM'$ with interpretation~$\calI'$ falsifies~$\calG$. We prove that~$\calM', \calI'(\sigma) \nvDash \phi$ whenever~$\sigma : \phi \in \calG$ by induction on the structure of~$\phi$. Recall that~$\calG$ does not contain formulas of the form~$\Diamond \psi$.
    \begin{compactitem}
        \item 
        If~$\sigma : \ell \in \calG$, then~$\calI(\sigma) \notin V'(p)$ in case $\ell = p$ and $\calI(\sigma) \in V'(p)$ in case $\ell = \overline{p}$. In both cases, $\calM', \calI'(\sigma) \nvDash \ell$. 

        \item        
If~$\sigma : \overline{\ell} \in \calG$, then~$\calM, \calI(\sigma) \nvDash  \overline{\ell}$ by~\eqref{Ap_base_literal}. If $\ell = p$, then~$\calI(\sigma) \in V'(p)$. By construction, also~$\calI'(\sigma) \in V_\calN(p)$. Note that~$\sigma : p \notin \calG$, otherwise Step~\ref{step_initial} would have been applied. Therefore~$\calI'(\sigma) \in V'(p)$ and thus~$\calM', \calI' \nvDash \overline{p}$. We leave the case that $\ell = \overline{p}$ to the reader. 

        \item
        If~$\sigma : \ell' \in \calG$ with~$\ell' \in \Lit \setminus \{\ell, \overline{\ell} \}$, then~$\calM, \calI(\sigma) \nvDash  \ell'$ by~\eqref{Ap_base_literal}, which immediately transfers to~$\calM'$ and~$\calI'$ by construction.

        \item 
        The cases for~$\land$ and~$\lor$ easily follow by saturation of~$\calG$. 

        \item 
        If~$\sigma : \Box \psi \in \calG$, we use the fact that~$\calG$ is saturated to see that~$\tau : \psi \in \calG$ for some~$\tau$ such that~$\calI'(\sigma)R\calI'(\tau)$. By induction hypothesis,~$\calM',\calI'(\tau) \nvDash \psi$ and therefore~$\calM',\calI'(\sigma) \nvDash \Box \psi$. 
    \end{compactitem}

\subsubsection*{BLUIP\eqref{BLUIP:3}, Step~\ref{step_saturated_dtrunk}.}
    Now suppose~$A_\ell(0, \Sigma_c;\Gamma)$ is obtained by Step~\ref{step_saturated_dtrunk}, i.e.,~the sequent consists of only the trunk and contains at least one diamond formula. Multiformula~$A_\ell(0, \Sigma_c;\Gamma)$ is calculated as in~\eqref{Ap_saturated_dtrunk} making use of the SDNF of $A_\ell(0, \Sigma_c ; \Gamma, [\{\psi \mid \Diamond \psi \in \Gamma \}]_{\bullet 1})$ presented in~\eqref{Ap_saturated_SDNFdtrunk}. By~$\calM, \calI \nvDash A_\ell(0,\Sigma_c;\Gamma)$, we know from~\eqref{Ap_saturated_dtrunk} that
    \begin{center} 
\adjustbox{max width = \textwidth}{    
$        \calM, \calI \nvDash \bigspdisj_{i=1}^{h} \big( \bullet : \Diamond \delta_i  \spconj \bullet : \gamma_{i} \big) \spdisj \bullet : \Box \bot$
  \quad or\quad
        $\calM, \calI \nvDash \bigspdisj_{\bullet : \ell' \in \Gamma, \ell' \in \Lit \setminus \{ \ell \}} \bullet : \ell' \spdisj \bullet : \Diamond \top$
    }
    \end{center}
    In the latter case,~$\calM$ is a one-world (irreflexive) model. We only need to define a new valuation on~$\calM$ to define~$\calM'$ as in Step~\ref{construction_model_valuation} and the proof proceeds the same way and is left to the reader. In particular, note that~$\Gamma$ has no boxed formulas as it only consists of the trunk.

    For the former case,~$\calM$ consists of more than one world. In particular, for each~$1 \leq i \leq h$, either~$\calM, \calI(\bullet) \nvDash \gamma_i$ or for all worlds~$v$ such that~$\calI(\bullet)Rv$ we have~$\calM, v \nvDash \delta_i$. For all such~$v$, consider the  interpretation~$\calJ_v : \{\bullet, \bullet 1 \} \rightarrow \calM$ defined by mapping label~$\bullet 1$ to world~$v$. By inspection of~\eqref{Ap_saturated_SDNFdtrunk}, observe that for all such~$v$ we have~$\calM, \calJ_v \nvDash A_\ell(0, \Sigma_c ; \Gamma, [\{\psi \mid \Diamond \psi \in \Gamma \}]_{\bullet 1})$. It will be sufficient to fix such a~$v$. Since~$A_\ell(0, \Sigma_c ; \Gamma, [\{\psi \mid \Diamond \psi \in \Gamma \}]_{\bullet 1})$ is sufficient by Lemma~\ref{lem:Ap_properties}\eqref{thm_item:dtrunk}, we can apply the induction hypothesis to obtain $(\calM', \calJ') \leqpos (\calM, \calJ_v)$ such that $\calM', \calJ' \nvDash \Gamma, [\{\psi \mid \Diamond \psi \in \Gamma \}]_{\bullet 1}$. Now let~$\calI'$ be~$\calJ'$ restricted to the domain~$\{ \bullet\}$. Hence,~$\calM', \calI' \nvDash \calG$. 
\qed
\end{proof}
\end{document}